\documentclass[11pt]{article}

\usepackage[T1]{fontenc}
\usepackage{lmodern}
\usepackage{microtype}
\usepackage{geometry}
\usepackage{amsmath,amssymb,amsthm,mathtools}
\usepackage{booktabs,tabularx,array}
\usepackage[]{hyperref}
\usepackage{xcolor}
\usepackage{appendix}
\usepackage{algorithm}
\usepackage{algorithmic}
\usepackage{dsfont}
\usepackage{xcolor}
\usepackage{enumitem}
\usepackage{times}
\usepackage{graphicx}

\newtheorem{theorem}{Theorem}[section]
\newtheorem{lemma}[theorem]{Lemma}
\newtheorem{proposition}[theorem]{Proposition}
\theoremstyle{definition}
\newtheorem{definition}[theorem]{Definition}
\theoremstyle{remark}

\newtheorem{claim}[theorem]{Claim}

\renewcommand{\algorithmicensure}{\textbf{Output:}}

\newcolumntype{Y}{>{\raggedright\arraybackslash}X}

\title{Online Fair Division: Pushing the Frontier of Approximate Proportionality}
\author{Yingjian Du\thanks{School of Computer Science and Technology, Shandong University, Qingdao, China. 202400130236@mail.sdu.edu.cn}\and Ankang Sun\thanks{School of Computer Science and Technology, Shandong University, Qingdao, China. ankang.sun@sdu.edu.cn}}
\date{}

\begin{document}
\maketitle

\begin{abstract}
Online fair division captures allocation problems in which indivisible
resources arrive over time and must be assigned before future resources
are known. Understanding what fairness remains achievable when allocation
decisions are immediate and irrevocable is a fundamental question in this
setting. We study deterministic online allocation among $n$ agents with nonnegative additive valuations, where the number of goods is unknown and the adversary can adapt to previous allocation decisions. 
We focus on proportionality up to one good (PROP1) and examine how advance future information affects the achievable guarantees.

Without additional future information, we give a deterministic algorithm that guarantees $\Omega(\frac{1}{\log(nm)})$-PROP1 after every round, where $m$ is the number of goods at termination. 
We complement this result by showing that, for every $n$ and sufficiently large $m$, every deterministic algorithm has an adaptive instance with $m$ goods on which its allocation has PROP1 approximation guarantee $O(\frac{\log \log m}{\log m})$.
Thus, when the number of agents is fixed, our upper and lower bounds on the competitive ratio differ by at most an $O(\log\log m)$ factor.
These results answer an open question in Choo et al. \cite{DBLP:journals/corr/abs-2508-03253} that asks whether a nontrivial deterministic approximation for PROP1 can be obtained. 

We also study the setting where the algorithm knows in advance the predictions of the maximum item value for every agent. 
When the predictions are accurate, we give a deterministic \(\frac{1}{2}\)-PROP1 algorithm, improving the \(\frac{1}{n}\) guarantee of Choo et al. \cite{DBLP:journals/corr/abs-2508-03253} to a constant.
We further establish an explicit upper bound below one on the competitive ratio, even for two agents with accurate predictions.
Finally, we give a single deterministic algorithm that guarantees \(\frac{1}{2}\)-PROP1 when predictions are accurate and \(\Omega(\frac{1}{\log (nm)})\)-PROP1 for arbitrary predictions.

\end{abstract}

\section{Introduction}
Fair division studies how scarce resources should be allocated among agents
with heterogeneous preferences in a fair manner
\cite{DBLP:books/daglib/0017730,DBLP:books/daglib/0017734}.
The offline problem, in which all goods and valuations are available before
any allocation is made, has been studied extensively
\cite{DBLP:journals/ai/AmanatidisABFLMVW23}.
In many practical applications, however, the resources are not available
simultaneously. Donations arrive over time and must be routed immediately to
charities, computational jobs must be assigned to servers upon release, and
opportunities such as advertising impressions or service requests must be
allocated when they become available
\cite{DBLP:conf/ijcai/AleksandrovAGW15,DBLP:conf/wine/BeiLL22,
DBLP:conf/nips/SpaehE23}.
These applications motivate \emph{online fair division}: indivisible goods
arrive sequentially, and each good must be allocated immediately and
irrevocably after the agents' values for that good are revealed
\cite{DBLP:conf/aaai/AleksandrovW20,
DBLP:conf/sigecom/BenadeKPP18,DBLP:conf/sigecom/ZengP20}.
The central difficulty is that an allocation that appears fair at the
current time may become highly unfair after additional goods arrive.

We take as our baseline the setting without additional information. There
are $n\geq 2$ agents with nonnegative additive valuations, and a finite
sequence of $m$ goods arrives online. The algorithm knows $n$, but it knows
neither the number $m$ of goods that will arrive nor the values of future
goods. Upon the arrival of a good, its value to every agent is revealed, and
the algorithm must assign it before seeing the next good. We allow the
adversary to be adaptive: the value vector of a future good may depend on
all previous allocation decisions. Fairness is evaluated ex post on the
final allocation when the adversary terminates the finite sequence.

Two canonical fairness notions are \emph{envy-freeness} and
\emph{proportionality}. The former requires that no agent envies another,
and the latter asks every agent to receive at least a $1/n$ fraction of her
value for all goods \cite{steinhaus1948problem}. Because goods are
indivisible, neither notion can always be satisfied, even with two agents
and a single good valued positively by both. This motivates the standard
relaxations of envy-freeness up to one good (EF1) and proportionality up to
one good (PROP1). Informally, EF1 allows envy to be eliminated by removing
one good from the envied bundle, whereas PROP1 allows an agent to reach her
proportional share by adding one good outside her bundle
\cite{DBLP:conf/sigecom/LiptonMMS04,
DBLP:conf/sigecom/ConitzerF017,
DBLP:journals/aamas/AzizCIW22,
DBLP:conf/aaai/BarmanK19}.
We use $\alpha$-PROP1 for the corresponding multiplicative relaxation and
refer to the worst-case factor guaranteed by an online algorithm as its
competitive ratio.

Recent work has revealed a sharp contrast between offline and online fair
division. In particular, EF1 allocations always exist in the offline
setting \cite{DBLP:conf/sigecom/LiptonMMS04}.
However, in the online setting
without additional information about future goods, no deterministic online
algorithm can guarantee any positive multiplicative approximation to EF1
\cite{DBLP:journals/corr/abs-2505-24503}.
%{\color{blue}\cite{DBLP:conf/ifaamas/MelissourgosP26} study an online fair division problem with a fixed number of goods.
%They show that for two agents with identical valuations, one can achieve $(\phi-1)$-EFX, while no $a$-EFX algorithm without predictions exists for any $a\in(\phi-1,1]$.}
For PROP1, the lower bound of
Benad\`e et al.~\cite{DBLP:conf/sigecom/BenadeKPP18} rules out exact PROP1 without additional information, but it does not quantify the best achievable approximation as a function of $n$ and $m$.
%{\color{blue}A natural relaxation of proportionality is proportionality up to $c$ goods(PROP$c$).
%Against an adaptive adversary, \cite{DBLP:journals/corr/abs-2605-19844} present a fully online algorithm that keeps the outcome PROP$c_t$ at any time, where $c_t \in O\left( \sqrt{\frac{t \log n}{n}} \right)$}
On the positive side, Neoh et al. \cite{DBLP:journals/corr/abs-2505-24503} show that exact PROP1 becomes
achievable when every agent's total value for the complete sequence is
known in advance. Thus, these results indicate that the approximability of
PROP1 depends on what information about future goods is available in
advance.

The work closest to ours is by Choo et al. \cite{DBLP:journals/corr/abs-2508-03253}, who also study multiplicative PROP1 approximations in this online model. 
Against adaptive adversaries, they show that three natural greedy algorithms can return allocations with arbitrarily small PROP1 factors.
Against a non-adaptive adversary, they show
that uniform random allocation achieves a
$\Theta(\frac{1}{\log\frac{n}{\delta}})$-PROP1 guarantee with probability at least $1-\delta$. 
They also introduce \emph{maximum item value} (MIV) predictions: before the process begins, the algorithm is given, for every
agent $i$, a prediction of the maximum value of a single good in the
realized instance. When these predictions are accurate, they give a
deterministic $\frac{1}{n}$-PROP1 algorithm against adaptive adversaries.

These results leave three questions. First, without additional information, what is the tight PROP1 guarantee for deterministic algorithms? 
Second, with accurate MIV predictions, can one obtain a constant guarantee, and can the guarantee be made arbitrarily close to one? Third, can a single deterministic algorithm exploit accurate MIV predictions while achieving a meaningful guarantee when those predictions are inaccurate?

In this work, we make progress on all three questions. Without additional information, we give an \(\Omega(\frac{1}{\log(nm)})\)-PROP1 deterministic algorithm and prove an upper bound of \(O(\frac{\log\log m}{\log m})\) for every fixed number of agents. With accurate MIV predictons, we obtain a \(\frac{1}{2}\)-PROP1 guarantee and establish an explicit impossibility bound below one, even for two agents. We then combine the two algorithmic guarantees and present one deterministic algorithm achieves \(\frac{1}{2}\)-PROP1 under accurate MIV predictions and \(\Omega(\frac{1}{\log(nm)})\)-PROP1 under arbitrary predictions.

\subsection{Our Contributions}
\label{subsec:contributions}

We study deterministic online algorithms for PROP1 and summarize our
results and contributions below. To state our results, we use the
following notation, which is defined formally in
Section~\ref{sec::pre}. For every agent $i$ and round $t$,
measured at the end of round $t$, let $V_i^t$ denote agent $i$'s total
value for all goods that have arrived by round $t$; let $B_i^t$ denote
her value for the goods allocated to her by round $t$; and let $U_i^t$
denote the maximum value she assigns to an arrived but unreceived good.

\paragraph{Without additional information: the algorithm.}
We give a deterministic online algorithm that is not told $m$ and guarantees \(\Omega(\frac{1}{\log(nm)})\)-PROP1 on every finite instance with $m$ goods. 
In other words, after every round \(t\), its allocation is \(\Omega(\frac{1}{\log(nt)})\)-PROP1. 
Thus, the guarantee remains valid whenever the adaptive adversary terminates the sequence. For every fixed number of agents, the terminal guarantee is \(\Omega(\frac{1}{\log m})\).

This result addresses the deterministic approximation question
posed by Choo et al. \cite{DBLP:journals/corr/abs-2508-03253}. Their work
shows that three natural greedy algorithms fail to guarantee any positive
competitive ratio against adaptive adversaries and obtains a logarithmic
guarantee through randomization against a non-adaptive adversary. Our
result shows that a nontrivial guarantee can also be obtained by a
deterministic algorithm against an adaptive adversary. The two positive
results are otherwise incomparable: their guarantee is independent of
$m$, but is probabilistic and assumes a non-adaptive input sequence,
whereas ours holds deterministically against adaptive inputs but decreases
with $m$.

At a high level, for every agent $i$, the algorithm maintains the largest
item value $H_i^t$ observed by round $t$, the number of times this largest
value has increased, and the difference between a target fraction of
$V_i^t$ and the value $B_i^t$ received by the agent. It uses a weighted
exponential potential to measure these differences after normalization by
$H_i^t$, together with an additional term that accounts for increases in
$H_i^t$. When a good arrives, the algorithm assigns it to an agent for whom
receiving the good produces the largest reduction in the potential. The
resulting upper bound on the potential bounds each normalized difference
by $O(\log(nt))$, which then gives the desired PROP1 guarantee.

\paragraph{Without additional information: an impossibility result.}
We complement the algorithm with an impossibility result. For every fixed $n\geq 2$, every sufficiently large $m$, and every deterministic online algorithm, there is an adaptive instance with $m$ goods on
which the allocation returned by the algorithm has PROP1 factor at most $\frac{4n\log\log m}{\log m}$.
The construction uses two agents with nonzero valuations and extends to general $n$ by adding dummy agents whose valuations are zero.
This result shows that no deterministic algorithm can guarantee a constant approximation ratio for PROP1 or a ratio that depends only on $n$.
When the number of agents $n$ is fixed, this upper bound and the established \(\Omega(\frac{1}{\log m})\) guarantee determine the optimal dependence on $m$ up to an \(O(\log\log m)\) factor.

% substantially strengthens the previous lower bound landscape.
% The impossibility of exact PROP1 implied by Benad\`e et
% al.~\cite{DBLP:conf/sigecom/BenadeKPP18} and discussed by Neoh et
% al.~\cite{DBLP:journals/corr/abs-2505-24503} rules out only competitive
% ratio one, whereas Choo et
% al.~\cite{DBLP:journals/corr/abs-2508-03253} establish competitive ratio
% zero only for three particular greedy algorithms. Our upper bound applies
% to every deterministic algorithm and rules out any approximation depends only on $n$. 
% Together with our
% $\Omega(1/\log m)$ guarantee, it determines the optimal dependence on $m$
% for every fixed $n$ up to an $O(\log\log m)$ multiplicative factor.

The construction fixes a parameter $R$ and, for every agent $i$ with
$U_i^t>0$, considers the normalized score
$ k_i^t=\frac{V_i^t-RB_i^t}{U_i^t}$.
As we will see, if $k_i^t>R$, then the PROP1 factor of the current allocation is smaller
than $\frac{n}{R}$. The adversary considers all finite continuations that are
consistent with the algorithm's decisions and subtracts a small penalty
proportional to the number of goods in each continuation. Starting from a
continuation whose penalized score for agent~$1$ is nearly the largest
possible, the adversary can choose subsequent goods that the algorithm
must allocate to agent~$1$. At the same time, agent~$2$ repeatedly misses
goods that are larger than all goods she previously missed, which
eventually makes $k_2^t$ exceed $R$. The penalty on the continuation length
bounds the total number of goods required. Choosing
$R=\Theta(\frac{\log m}{\log \log m})$ gives the stated upper bound.

\paragraph{With MIV information: the algorithm.}
In the setting with MIV information, before the first good arrives, the algorithm is given the $p_i$ of every agent $i$, where $p_i$ is the predicted maximum value of a single good for agent $i$.
This information specifies only the maximum single good value of each agent,
but reveals neither the number of goods, the agents' total values, nor
the arrival order.
When the predictions are accurate, we give a deterministic algorithm that guarantees \(\frac{1}{2}\)-PROP1, independent of both \(n\) and $m$.
This improves the \(1/n\) guarantee of Choo et al. \cite{DBLP:journals/corr/abs-2508-03253} to a constant.
This result shows that even limited advance information, namely, the maximum item value of every agent, can substantially improve the worst-case fairness guarantee.
%, without revealing the future valuation vectors or the agents' total values.

Informally, the algorithm normalizes the value of every arriving good
to agent $i$ by $p_i$. It maintains a normalized
shortfall from one half of the proportional share, together with one of
three states recording whether a good of value $p_i$ has appeared and
whether such a good has been allocated to another agent. These quantities
define a remaining margin $q_i^t$ for agent $i$ at the end of round $t$.
Let $ N^+:=\{i\in N:p_i>0\}$
be the set of agents with positive $p_i$'s. The algorithm uses the potential $\Psi^t:=\sum_{i\in N^+}\frac{1}{q_i^t}$,
so that an agent with a small remaining margin receives higher priority.
For each arriving good, the algorithm compares how the potential would
change under the possible allocation choices and selects a recipient that
keeps the potential from increasing. This ensures that
$q_i^t>0$ for every $i\in N^+$ and every round $t$. At termination, a
good of value $p_i$ must have appeared for every $i\in N^+$. Such a good
is either contained in $A_i$ or contributes to $U_i$, and the positivity
of $q_i^m$ then implies the desired approximation guarantee.

\paragraph{With MIV information: an impossibility result}
We complement the algorithm with an impossibility result. As observed by
Choo et al.~\cite{DBLP:journals/corr/abs-2508-03253}, accurate maximum item value predictions do not suffice to guarantee 1-PROP1, while there is no known constant gap. 
We strengthen this impossibility by establishing an explicit multiplicative gap below one.
Even for two agents with $p_1=p_2=1$, no deterministic algorithm can guarantee a competitive ratio greater than $1-\frac{1}{64\cdot 518}$.
The construction uses at most \(518\) goods, with every item value between \(\frac{1}{64}\) and \(1\).
The predicted maximum item value is fixed before allocation begins and is attained in the realized sequence. 
Thus, the construction does not rely on uncertainty about whether a good exceeding an announced maximum will arrive later. 
Instead, it uses the values and adaptive arrival order of goods within the announced bounds.

% every deterministic algorithm can be forced,
% using at most $518$ goods with values between $1/64$ and $1$, to return
% an allocation whose PROP1 factor is at most
% $
%     1-\frac{1}{64\cdot 518}.
% $

% In particular, the adversary keeps the complete MIV vector fixed from the
% beginning and creates the difficulty through the values and adaptive order
% of the remaining goods. This addresses an obstacle identified by Choo et
% al.~\cite{DBLP:journals/corr/abs-2508-03253}. That is, many online lower bounds depend on uncertainty about whether a much larger good will arrive later.
% Such an argument cannot be used directly with exact MIV information,
% because the announced maximum value must be preserved and must occur in
% the realized sequence.

The construction considers two agents and defines
$
    s_i^t=2(B_i^t+U_i^t)-V_i^t.
$
For two agents, agent $i$ satisfies PROP1 at round $t$ exactly when
$s_i^t\geq 0$. The construction has three phases. It first uses three
common goods with value one each to ensure that both agents have an unreceived good
of value one, unless the desired upper bound has already been obtained.
It then uses adaptive two-good blocks, each of which reduces
$s_1^t+s_2^t$ by $1/64$, until one agent has only a small remaining
margin. Finally, a five-good block forces the first four goods to that
agent and uses the last good to make $s_i^t$ negative for one of the
agents. 
% By choose the values properly, we obtain the desired multiplicative gap, which is below one.

\paragraph{An algorithm with consistency and robustness.} Finally, we are concerned with the case when the predictions are inaccurate and present a deterministic algorithm that leverages predictions to improve the worst-case guarantees when the predictions are accurate (consistency), while also providing a performance guarantee when the predictions fail (robustness).
In particular, we give one deterministic algorithm that guarantees \(\frac{1}{2}\)-PROP1 when predictions are accurate, and guarantees $\Omega(\frac{1}{\log (nm)})$-PROP1 when predictions fail.
At a high level, the algorithm combines the potentials of the above two established algorithms with a proper refinement. The algorithm then allocates each arriving good to the agent that minimizes the resulting combined potential.

These results distinguish the fairness guarantees available under different forms of advance information. 
Knowing every agent’s total value for the complete sequence suffices for exact PROP1 \cite{DBLP:journals/corr/abs-2505-24503}. 
Knowing only each agent’s maximum item value permits a constant guarantee, but does not suffice for exact or arbitrarily close to exact PROP1. 
Without additional information, the best deterministic guarantee decreases with $m$ for every fixed $n$, while an \(\Omega(\frac{1}{\log m})\) guarantee remains achievable. 
Our last algorithm further shows that the benefit of accurate maximum item value predictions can be obtained while retaining the logarithmic guarantee when those predictions are inaccurate.

% Although both of our positive results use potential functions to guide the
% allocation of each arriving good, they use different normalization values.
% Without additional information, the algorithm uses the largest item value
% observed so far, which may increase during the process. With MIV
% information, the algorithm uses the fixed value $p_i$ known before the first round. Taken together, our results establish a three level
% information hierarchy for online proportionality. Knowing every agent's
% total value for the complete sequence suffices for exact PROP1
% \cite{DBLP:journals/corr/abs-2505-24503}. Knowing only each agent's exact
% MIV does not suffice for exact PROP1, but it does suffice for a
% $1/2$-PROP1 guarantee independent of the number of agents. Without either
% form of information, the achievable competitive ratio must decrease with
% $m$ for fixed $n$, although an $\Omega(1/\log m)$ guarantee remains
% achievable.

\subsection{Technical Overview}
\label{subsec:techniques}
We first describe the ideas behind the algorithms for the without information and accurate maximum item value predictions settings, followed by the two impossibility constructions.
The algorithm that provides both consistency and robustness combines modified versions of the two algorithmic potentials.

\paragraph{Algorithms.}
Both algorithms track a proportionality debt for each agent: the difference between a target fraction of
her value for the goods seen so far and the value she has received. This debt is normalized by an available item value
scale. For each arriving good, the algorithm computes the exact potential reduction obtained by assigning the good
to an agent rather than letting her miss it, and chooses an agent with maximum reduction. A one step potential
comparison then shows that the reduction for the recipient covers the potential increases of the agents who miss the
good. 
Without information about future goods, the scale is the running maximum and may change when a new record
arrives. The algorithm uses a weighted exponential potential so that a bound on this potential gives a logarithmic bound on normalized debt, together with a telescoping reserve that accounts for changes in the running maximum.
With accurate MIV information, the value $p_i$ gives a fixed scale.
The algorithm records whether a maximum valued
good is available for the one good correction, and a reciprocal potential keeps the resulting safety slack positive. In
both cases, the potential controls normalized proportionality debt, and the largest relevant good converts the debt
bound into a PROP1 guarantee.
The last algorithm combines these approaches and use a single allocation rule that controls the combined potential.

\paragraph{Counterexamples.}
Both lower bounds first describe the fairness status of each active agent by a small set of variables and derive their exact changes when the current good is received or missed. 
The adversary then chooses values adaptively so that one allocation choice either violates the target guarantee or enters a state from which a violation can be forced. 
This restricts the algorithm's choice and allows the adversary to steer the allocation over several rounds, using one agent to route goods while reducing the fairness margin of another. 
The two constructions implement this idea differently because the available value scales are different. 
Without information, the largest missed value is an endogenous scale that may increase over time, while with perfect MIV information, the scale is fixed in advance. 
The common idea here is to combine a low dimensional state description, potential changes, and a region that the algorithm cannot safely enter. 
This may also be useful when designing adaptive lower bounds for other online fairness problems.

\subsection{Further Related Work}
\label{subsec:related-work}

\paragraph{Online fair division.}
Online fair division was initially studied under restricted valuation
domains, motivated in part by food-bank allocation
\cite{DBLP:conf/ijcai/AleksandrovAGW15}; see \cite{DBLP:conf/aaai/AleksandrovW20} for a
survey. Subsequent work studies cumulative envy, fairness--efficiency
tradeoffs, and the effect of allowing previous allocations to be revised
\cite{DBLP:conf/sigecom/BenadeKPP18,DBLP:conf/ijcai/HePPZ19,DBLP:conf/sigecom/ZengP20}. More recent work
investigates ex-post multiplicative guarantees for standard fairness
notions. Zhou et al. \cite{DBLP:conf/icml/0002B023} study MMS under normalized valuations, while
Neoh et al. \cite{DBLP:journals/corr/abs-2505-24503} compare no information, normalization
information, and richer frequency predictions. Choo et al. \cite{DBLP:journals/corr/abs-2508-03253}
focus on approximate PROP1 under adaptive and non-adaptive adversaries and introduce MIV predictions.
Kahana et al. \cite{DBLP:journals/corr/abs-2605-19844} give a deterministic online algorithm that bounds each agent's deficit by $c_t$ times the value of her largest unreceived good after $t$ rounds.
For a fixed number of agents, $c_t=O\left( \sqrt{t}\right)$, which implies an $\Omega\left({m^{-\frac{1}{2}}}\right)$-PROP1 guarantee.
Their algorithm provides an $m$ dependent deterministic guarantee without additional information against an adaptive adversary.
Very recent work by Chen and Tan \cite{chen2026competitive} develops a broader competitive analysis
framework across several fairness notions and information structure; its
normalization by the best offline fairness factor differs from the direct
$\alpha$-PROP1 benchmark used here.

\paragraph{Fair division with predictions and advice.}
Our MIV model is related to the learning-augmented analysis of online algorithms, which augments worst-case online input with predictions and seeks guarantees that exploit accurate advice without abandoning rigorous
worst-case analysis
\cite{DBLP:conf/nips/PurohitSK18,DBLP:journals/jacm/LykourisV21,
DBLP:journals/cacm/MitzenmacherV22}. Prediction augmented allocation has been studied for welfare, scheduling, matching, and mechanism-design objectives
\cite{DBLP:conf/soda/LattanziLMV20,DBLP:conf/innovations/BalkanskiGT23,
DBLP:conf/nips/SpaehE23,DBLP:conf/nips/CohenEEV24}. In online fair division, prior prediction
models include entire valuation vectors, total-value information, frequency predictions, and MIV
predictions \cite{DBLP:conf/ifaamas/MelissourgosP26, DBLP:journals/corr/abs-2505-24503,DBLP:journals/corr/abs-2508-03253}. Our purpose is
not to learn the values of arriving goods, that is, the entire current value vector
is observed, but to understand how a very small amount of global scale information changes the worst-case approximability of an ex-post fairness
benchmark.

\paragraph{Offline proportionality relaxations.}
The notion of PROP1 is a widely studied relaxation of the canonical proportionality for indivisible items
\cite{DBLP:conf/sigecom/ConitzerF017,DBLP:journals/aamas/AzizCIW22,aziz2020polynomial,
DBLP:conf/aaai/BarmanK19}. In the offline model, PROP1 follows from several
standard allocation procedures and stronger fairness guarantees. The online setting is different: irrevocability prevents balancing after
future values are known, and the adversary can choose when the process
terminates. Our results isolate the resulting price of uncertainty, as well as the extent to which one scalar of exact scale information per agent can
mitigate it.

\paragraph{Concurrent and Independent Work.}
Independently and concurrently to our work, Neoh and Teh~\cite{neoh2026closing} show that, without additional information, no deterministic or randomized online algorithm can guarantee a positive PROP1 competitive ratio only on the number of agents $n$ against an adaptive adversary.
In settings with accurate MIV information, their potential is different from ours, and they give a deterministic algorithm with higher competitive ratio than ours.
Given an upper bound $\kappa\in[2,n]$ on the number of agents who value any single good positively, their guarantee improves to $\max\{\frac{19}{30},\frac{n}{n+\kappa}\}$-PROP1.
With accurate MIV predictions and a known number of goods $m$, they give a deterministic algorithm that simultaneously guarantees $\frac{19}{30}$-PROP1 and maximum additive envy of $O(\log n+\sqrt{\frac{m\log n}{n}})$, after normalizing each agent's values by her MIV.
For inaccurate MIV predictions, they obtain a positive constant approximation under any fixed, one-sided error bound below one.
They also analyze randomized allocation rules against a non-adaptive adversary.

Our work complements their result of impossibility in online settings without additional information by quantifying the dependence on the number of goods $m$.
Furthermore, we give a deterministic online algorithm guarantees $\Omega\left(\frac{1}{\log(nm)}\right)$ on every instance with $m$ goods.
In online settings with accurate MIV predictions and an adaptive adversary, we give a deterministic algorithm guarantees $\frac{1}{2}$-PROP1, and prove the impossibility result that no deterministic algorithm can guarantee a competitive ratio greater than $1-\frac{1}{64\cdot 518}$.
Our consistency and robustness result provides a deterministic algorithm that achieves $\frac{1}{2}$-PROP1 under exact predictions and $\Omega(\frac{1}{\log(nm)})$-PROP1 under arbitrary nonnegative predictions, without requiring any bound on the prediction error.

% \paragraph{Paper organization.}
% Section~\ref{sec:preliminaries} introduces the model and the multiplicative
% PROP1 benchmark. Sections~\ref{sec:noinfo-lower} and
% \ref{sec:noinfo-algorithm} present, respectively, the zero horizon-uniform
% lower bound and the logarithmic no-information algorithm.
% Sections~\ref{sec:miv-lower} and~\ref{sec:miv-algorithm} establish the
% quantitative impossibility and the $1/2$-PROP1 algorithm under perfect MIV
% information. We conclude with open questions, including the optimal
% dependence on $m$ without information and the optimal constant under
% perfect MIV.

\section{Preliminaries}\label{sec::pre}

For any $k\in \mathbb{N}_+$, let $[k]=\{1,\ldots,k\}$. There are a set $N=\{1,\ldots,n\}$ of $n$ agents and a finite sequence of indivisible goods a set $M=\{g_1,\ldots,g_m\}$ of $m$ indivisible goods to be allocated to these agents.
We study the setting where the number of agents is known in advance and goods arrive \emph{online}, one at a time. 
In particular, $m$ is unknown, and moreover, when a good arrives, it must be allocated immediately and irrevocably to one of the agents.
Each agent $i$ is associated with an \emph{additive} valuation function $v_i: 2^M\rightarrow \mathbb{R}_{\geq 0}$.
The additivity of $v_i$ gives that for any $S\subseteq M$, $v_i(S)=\sum_{g\in S} v_i(\{g \})$. For simplicity, instead of $v_i(\{g\})$, we write $v_i(g)$ hereafter.

Since goods arrive online, we label the goods in their arrival order. For any $t\in \mathbb{N}_+$, let $g_t$ be the good that arrives at round $t$.
An allocation $A=(A_1,\ldots,A_n)$ is an $n$-partition of $M$, such that for any $i\neq j$, $A_i\cap A_j =\varnothing$ and $\bigcup_{i\in [n]} A_i = M$. Each $A_i$ refers to the set of goods or the \emph{bundle} allocated to agent $i$.
For any round $t$, let $M^t=\{g_1,\ldots,g_t\}$ denote the set of arrived goods. 
For any round $t$ and any $i\in [n]$, let $A^t_i$ denote the bundle already allocated to agent $i$ after the allocation of the good $g_t$.

For the fairness notion, we study the proportionality and focus on its relaxation. The proportional share of each agent $i$ is $\frac{v_i(M)}{n}$, and a proportional (PROP) allocation ensures that each agent receives her proportional share.
It is known that PROP allocations do not always exists, even in the offline setting of allocating one good to two agents.
We instead study the relaxation of proportionality, so called proportional up to one good (PROP1).

\begin{definition}[$\alpha$-PROP1] \label{def:alpha-prop1} For any $\alpha\in[0,1]$, the allocation $A=(A_1,\ldots,A_n)$ is \emph{$\alpha$-proportional up to one good} ($\alpha$-PROP1) if for any $i\in N$, either $v_i(A_i)\geq \frac{v_i(M)}{n}$, or there exists a good $g\in M\setminus A_i$ such that $v_i(A_i\cup\{g\}) \geq \alpha \cdot \frac{v_i(M)}{n}$.
When $\alpha=1$, we simply say that $A$ is PROP1. \end{definition}

To simplify notations, for any agent $i\in N$ and any $t\in\{0,\ldots,m\}$, we define the following quantities at the end of
round $t$, after good $g_t$ has been allocated:
$
    V_i^t
    = \sum_{\ell=1}^t v_i(g_\ell),
$
which is agent $i$'s total value for all goods that have arrived by
round $t$;
$
    B_i^t
    = \sum_{g_\ell\in A_i^t} v_i(g_\ell),
$
which is agent $i$'s value for the goods allocated to her by round
$t$; and
$
    U_i^t
    = \max
        \{v_i(g_\ell):g_\ell\in M^t\setminus A_i^t\},
$
which is the maximum value that agent $i$ assigns to a good that has
arrived by round $t$ but has not been allocated to her. We set
$A_i^0=\emptyset$ for all $i$, and hence $V_i^0=B_i^0=U_i^0=0$. 
These notations will be used throughout the
paper.

% To simplify notations, for any agent $i$ and round $t$, we define 
% $V^t_i=\sum_{\ell \leq t} v_i(g_\ell) $ the total value of all arrived goods until round $t$ for agent $i$;
% define $B^t_i=\sum_{\ell\leq t: g_\ell \in A^t_i} v_i(g_\ell)$ the total value of goods allocated to agent $i$ until round $t$;
% define $U^t_i=\max_{g_\ell\notin A_i^t:\ell \leq t} v_i(g_\ell)$ the maximum value of a good not allocated to agent $i$.
% These three notations will be used throughout the paper.

Since the number of goods, their values, and their arrival order are
unknown, the online allocation process can be viewed as an interaction
between an adaptive adversary and an online algorithm. After observing
the history through round $t-1$, including all previous allocation
decisions, the adversary either terminates the sequence or reveals a
new good $g_t$ together with its value $v_i(g_t)$ for every agent
$i\in N$. The algorithm must then allocate $g_t$ immediately and
irrevocably. 
For any $\alpha>0$, an algorithm guarantees \emph{competitive ratio} at least $\alpha$ if its terminal allocation is $\alpha$-PROP1 on every admissible instance.

We consider two information settings: \emph{without additional
information} and \emph{with maximum item value information}.
Without additional information, the algorithm receives no advance
information about the number, arrival order, or values of the goods
beyond what is revealed in the current round.
With MIV information, before the first round the algorithm is given $p_1,\ldots,p_n$,
where $p_i=\max_{\ell\in[m]} v_i(g_\ell)$ is the maximum value that agent $i$ assigns to any single good in the entire realized sequence. Thus, the adversary must generate a finite sequence satisfying $v_i(g_\ell)\leq p_i$ for every agent $i$ and every good $g_\ell$, and, whenever $p_i>0$, at least one good must satisfy $v_i(g_\ell)=p_i$. 
In particular, if $p_i=0$, then $v_i(g_\ell)=0$ for every good $g_\ell$. 
Apart from the vector $p_i$'s, the algorithm receives no information about the number, arrival
order, or values of the goods.

% As the number of goods $m$ and their arrival order are unknown, the problem can indeed be viewed as a game between an adversary and an online algorithm. 
% We consider the standard adaptive models. An adaptive adversary observes the algorithm’s past
% decisions and, at each round $t$, decides whether to terminate the sequence or introduce a new good $g_t$ together with its values for every agent.
% We study two different type of information settings, namely, \emph{without additional information} and \emph{with maximin item value (MIV) information}.
% In the setting without additional information, there is no extra information available to the algorithm, while in the setting with MIV information, before the first round, the algorithm knows $p_1,\ldots,p_n$, where $p_i$ is the maximum value of an agent for all future goods. In particular, at termination, the maximum value of a single item for agent $i$ must be $p_i$.

% We consider different type of information settings. In the \emph{fully online} setting, the algorithm does not receive extra information.
% In the \emph{perfect Maximum Item Value} prediction setting, before the first round, the algorithm is told $v_i^{\max}=\max_{g\in M} v_i(g)$ for every $i\in  [n]$.

% {\color{red}[Miss definitions about the notation $V^t_i, B^t_i,U^t_i$, and adaptive adversary]}

\section{Without Additional Information}\label{sec::no-info}
\subsection{Without Additional Information: A Deterministic Algorithm}
We first consider the setting without additional information. In this
section, we present a deterministic online algorithm that guarantees $\Omega(\frac{1}{\log nm})$-PROP1 on every
finite instance with $m$ goods. For ease of presentation, define function
$
    g(r)=\frac{1}{r(r+1)}
$ for every $r\in\mathbb{N}_+$,
and let $\theta=\frac{1}{4n}$ and $\gamma=\frac{3}{4}$. For any agent $i\in N$ and round $t\in[m]$, let
$
    x_i^t=v_i(g_t),
$
and $x^t=(x^t_1,\ldots,x^t_n)$.

We now explain the high level idea of the algorithm. For every agent $i$ and
round $t$, the algorithm maintains $H_i^t$, the largest value that agent
$i$ assigns to any good that has arrived by round $t$, and $\tau_i^t$, the
number of times that this largest observed value has strictly increased
by round $t$. The first positive value is counted as the first such
increase. The algorithm also maintains a quantity $D_i^t$. We later prove
that $
    D_i^t=\theta V_i^t-B_i^t$.
Thus, $D_i^t$ measures the difference between the target
$\theta V_i^t$ and the value that agent $i$ has received by round $t$.

When good $g_t$ arrives, the algorithm first updates $H_i^t$ and
$\tau_i^t$ for every agent. It then computes the score $\Delta_i^t$
defined in Line~\ref{alg-no-info::line::Delta}. For an agent with
$H_i^t>0$, this score is the difference between the agent's
exponential contribution to the potential when she does not receive
$g_t$ and her contribution when she receives $g_t$. Therefore, a larger
value of $\Delta_i^t$ means that assigning $g_t$ to agent $i$ avoids a
larger increase in the potential. The algorithm allocates $g_t$ to an
agent with maximum $\Delta_i^t$ and then updates $D_i^t$ for all agents.
The formal description is given in
Algorithm~\ref{alg:no-informaton}.

% Let us explain the high-level idea of the algorithm. The algorithm records $H^t_i$ the maximum value of single (arrived) good for agent $i$ at round $t$ and $\tau^t_i$ the number of times that $H$ is increased until round $t$.
% When $g_t$ arrives, it updates $H^t_i, \tau^t_i$ and then computes $\Delta^t_i$ (Line~\ref{alg-no-info::line::Delta}). The algorithm then allocates $g_t$ to the agent achieving the maximum $\Delta^t_i$ and update $D^t_i$ for all $i$. Indeed, $D^t_i$ is the difference between $1/(4n)$ fraction of her total value of goods arrived at round $t$ and $i$'s received value, which will be proved later on. 
% The higher the $D^t_i$ is, agent $i$ is farther from her proportionality benchmark.

% Term $\Delta^t_i$ represents how dangerous agent $i$ is if not receiving $g_t$. It increases with the increase of $D^{t-1}/H^t_i$ and of $\exp(\gamma \theta z^t_i) -\exp(-\gamma(1-\theta)z^t_i)$, where the former is the normalized difference the refined benchmark and $i$'s current value and the latter captures how useful of $g_t$ to agent $i$.
% Thus, the algorithm allocates $g_t$ to the agent with the maximum $\Delta^t_i$, i.e., the agent needs $g_t$ the most. 
% The formal description of the algorithm is introduced in Algorithm~\ref{alg:no-informaton}

\begin{algorithm}[th]	\caption{Deterministic allocation without additional information}
	\label{alg:no-informaton}
 \renewcommand{\algorithmicensure}{\textbf{Output:}}
	\begin{algorithmic}[1]
        \REQUIRE A fixed deterministic ordering of the agents for tie-breaking.
        \STATE Initialize $D^0_i=H^0_i=\tau^0_i=0$ and the empty allocation $A=(A_1,\ldots,A_n)$ with $A_i=\varnothing$ for all $i$.
        \WHILE{a new good $g_t$ arrives}
        \STATE Observe $x^t=(x^t_1,\ldots,x^t_n)$.
        \STATE For each $i\in [n]$, update $H^t_i\gets \max\{H^{t-1}_i,x^t_i\}$ and $\tau^t_i\gets \tau^{t-1}_i + \mathds{1}_{\{x^t_i>H^{t-1}_i\} }$.\label{alg-no-info::line::H-and-tau}
        \FOR{each $i\in [n]$}
        \IF{$H^t_i=0$}
        \STATE Define $z^t_i=0$ and $\Delta^t_i=0$.
        \ELSE
        \STATE Define $z^t_i=x^t_i/H^t_i$ and 
        $
        \Delta^t_i = g(\tau^t_i)\exp\left(\gamma \frac{D^{t-1}_i}{H^t_i}\right)\cdot \left[\exp(\gamma \theta z^t_i) -\exp(-\gamma(1-\theta)z^t_i)\right].
        $\label{alg-no-info::line::Delta}
        \ENDIF
        \ENDFOR
        \STATE Allocate $g_t$ to agent $i_t \in \arg\max_{i\in [n]} \Delta ^t_i$, i.e., $A_{i_t}\gets A_{i_t}\cup \{g_t\}$. Break ties according to the fixed deterministic ordering.
        \STATE Update: for each $j\neq i_t$,  $D_j^t \gets D^{t-1}_j + \theta x^t_j$, and for agent $i_t$, $D_{i_t}^t \gets D^{t-1}_{i_t} - (1-\theta)x^t_{i_t}$.
        \ENDWHILE
        \ENSURE Allocation $A$.
	\end{algorithmic}
\end{algorithm}

We next analyze Algorithm~\ref{alg:no-informaton}. The transition from
the potential at the end of round $t-1$ to the potential at the end of
round $t$ is divided into two steps. First, after $g_t$ is revealed,
$H_i^t$ and $\tau_i^t$ are updated, while the allocation has not yet been
made and the value of $D_i$ remains $D_i^{t-1}$. We represent this
intermediate state by a \emph{pre-allocation} potential. Second, after $g_t$ is
allocated, the values $D_i^t$ are updated, giving the \emph{post-allocation}
potential.

For every round $t\geq 1$, define the post-allocation potential by
\[
\Phi^t:=\sum_{i: H_i^t>0} g(\tau^t_i)\exp\left(\gamma \frac{D^{t}_i}{H^t_i}\right) + \sum_{i=1}^n\frac{1}{\tau^t_i+1},
\]
and \emph{pre-allocation} potential by
\[
\widehat{\Phi}^t:=\sum_{i: H_i^t>0} g(\tau^t_i)\exp\left(\gamma \frac{D^{t-1}_i}{H^t_i}\right) + \sum_{i=1}^n\frac{1}{\tau^t_i+1}.
\]
Thus, $\Phi^t$ and $\widehat{\Phi}^t$ use the same values of $H_i^t$
and $\tau_i^t$. They differ only in whether the numerator in the
exponential term is $D_i^t$ or $D_i^{t-1}$.

At $t=0$, the exponential sum is empty and
$
    \Phi^0=n.
$
We first prove that updating $H_i^t$ and $\tau_i^t$ does not increase
the potential, and then prove that the allocation and the update of
$D_i^t$ do not increase it. Consequently, $\Phi^t\leq n$ after every
round. This bound will yield a lower bound on the achieved PROP1 factor.

% Informally, $\Phi^t$ is computed after allocating $g_t$, while $\widehat{\Phi}^t$ is computed before allocating $g_t$ but after the update of $H^t_i$ and $\tau^t_i$. 
% This makes that the only difference between $\Phi^t$ and $\widehat{\Phi}^t$ is $D^t_i$ and $D^{t-1}_i$, the nominator on the exponential term.

% To bound the competitive ratio, we will prove that the post-allocation potential never increases. Initially, since $\tau^i_0=0$ and $H^0_i=0$ for all $i$, we have $\Phi^0=\widehat{\Phi}^0=n$. The monotonicity of the potential gives the upper bound of the post-allocation potential at every possible termination, which then bounds the approximation factor.

\begin{lemma}\label{lem::pre-and-post}
    For any round $t\geq 1$, $\widehat{\Phi}^t \leq \Phi^{t-1}$.
\end{lemma}
\begin{proof}
    It suffices to prove that for each agent $i$, her contribution to $\widehat{\Phi}^t$ is at most her contribution to $\Phi^{t-1}$. We split the proof by distinguishing $x^t_i\leq H^{t-1}_i$ or not.

    If $x^t_i\leq H^{t-1}_i$, then by Line~\ref{alg-no-info::line::H-and-tau}, we have $H^t_i=H^{t-1}_i$ and $\tau^t_i=\tau^{t-1}_i$.
    By the definitions, one can verify that the contributions of agent $i$ to $\widehat{\Phi}^t$ and $\Phi^{t-1}$ are equal.

    In the remaining case, where $x^t_i> H^{t-1}_i$, we distinguish between $\tau^{t-1}_i=0$ and  $\tau^{t-1}_i\geq 1$.
    If $\tau^{t-1}_i=0$, then $H^{t-1}_i=0$, so that $H^\ell_i=x_i^\ell = 0$ for all $\ell \leq t-1$. Hence, agent $i$'s contribution to $\Phi^{t-1}$ is $1/(\tau^{t-1}+1)=1$.
    After the arrival of $g_t$, we compute $\widehat{\Phi}^t$, and the contribution of agent $i$ is
    \[
    g(\tau^t_i)\exp(\gamma \frac{D_i^{t-1}}{H^t_i}) + \frac{1}{\tau^t+1} = g(1)e^0+\frac{1}{2} = 1,
    \]
    which equals agent $i$'s contribution to $\Phi^{t-1}$.

    If $\tau^{t-1}_i\geq 1$, define
    \[
    \widehat{\Phi}^t_i=g(\tau^t_i)\exp(\gamma \frac{D^{t-1}}{H^t_i}) + \frac{1}{\tau^t+1} \text{ and } \Phi^{t-1}_i=g(\tau^{t-1}_i)\exp(\gamma \frac{D^{t-1}}{H^{t-1}_i}) + \frac{1}{\tau^{t-1}+1},
    \]
    which refer to $i$'s contribution to $\widehat{\Phi}^t$ and $\Phi^{t-1}$, respectively.
    When $D^{t-1}\geq 0$, we have $D^{t-1}/H^t_i \leq D^{t-1}/H^{t-1}_i$, since $H^t_i=x^t_i>H^{t-1}_i$.
    Since $\tau^t_i=\tau^{t-1}_i+1$, we have $g(\tau^t_i)< g(\tau^{t-1}_i)$ and $1/(\tau^t+1) \leq 1/(\tau^{t-1}+1)$. Combining everything, we have $\widehat{\Phi}^t_i < \Phi^{t-1}_i$.
    On the other hand, when $D^{t-1}<0$, we have
    \[
    \widehat{\Phi}^t_i<g(\tau^t_i)+\frac{1}{\tau^t+1} = \frac{1}{\tau^t_i} - \frac{1}{\tau^t_i+1} + \frac{1}{\tau^t+1} \leq \Phi^{t-1},
    \]
    where the last inequality transition is due to that $\tau^t_i=\tau^{t-1}_i+1$ and the exponential contribution is positive.

    Therefore, by summing over all agents' contribution, we can claim $\widehat{\Phi}^t \leq \Phi^{t-1}$.
\end{proof}

Next we compare $\Phi^t$ and $\widehat{\Phi}^t$. Let us first write the expression of their difference:
\[
\Phi^t-\widehat{\Phi}^t=\sum_{i:H^t_i>0} g(\tau^t_i) \left(\exp(\gamma\frac{D^t_i}{H^t_i}) -  \exp(\gamma\frac{D^{t-1}_i}{H^t_i})  \right).
\]
Recall that good $g_t$ is allocated to $i_t$. By substituting that $D^t_i=D^{t-1}_i+\theta x^t_i$ for all $i\neq i_t$ and $D^t_{i_t}=D^{t-1}_{i_t}-(1-\theta)x^t_{i_t}$, we can then rearrange the difference as
\[
\Phi^t-\widehat{\Phi}^t=\sum_{i=1}^n g(\tau^t_i)\exp\left(\gamma \frac{D^{t-1}_i}{H^t_i} \right) \left[\exp(\gamma \theta z^t_i)-1 \right] - \Delta^t_{i_t},
\]
where $\Delta^t_{i_t}$ is defined in Line~\ref{alg-no-info::line::Delta} and it captures the difference of $i_t$'s contribution between receiving $g_t$ and losing $g_t$.
For ease of presentation, for any $i$ and $t$, define 
\[
J^t_i:=g(\tau^t_i)\exp\left(\gamma \frac{D^{t-1}_i}{H^t_i} \right) \left[\exp(\gamma \theta z^t_i)-1 \right],
\]
and hence, $\Phi^t-\widehat{\Phi}^t=\sum_{i=1}^n J^t_i- \Delta^t_{i_t}$.

\begin{lemma}\label{lem::J-and-Delta}
    For any agent $i$ and round $t$, $J^t_i \leq \frac{3}{8n}\cdot\Delta^t_i$.
\end{lemma}
\begin{proof}
    Recall that $\Delta^t_i = g(\tau^t_i)\exp\left(\gamma \frac{D^{t-1}_i}{H^t_i}\right)\cdot \left[\exp(\gamma \theta z^t_i) -\exp(-\gamma(1-\theta)z^t_i)\right]$. 
    If $z^t_i=0$, we have $J^t_i=\Delta^t_i=0$ and the statement holds trivially.

    We now focus on $z^t_i>0$. By definition, we have
    \[
    \frac{J^t_i}{\Delta^t_i} = \frac{e^{\gamma \theta z^t_i}-1}{e^{\gamma \theta z^t_i} - e ^{-\gamma(1-\theta)z^t_i}}, \text{which implies } \frac{J^t_i}{\Delta^t_i} = \frac{1-e^{-\gamma \theta z^t_i}}{1-e^{-\gamma z^t_i}},
    \]
    where the right hand side equation is achieved by multiplying the numerator and denominator of the left hand side fraction by $e^{-\gamma \theta z^t_i}$.
    Due to that $1-e^{-y} \leq y$ for all $y\geq 0$, we have $1-e^{-\gamma \theta z^t_i} \leq \gamma \theta z^t_i$.
    Thus, $J^t_i/\Delta^t_i \leq \gamma \theta z^t_i/(1-e^{-\gamma z^t_i})$.
    Notice that function $f(y)=1-e^{-y}$ is concave on $[0,\infty)$ and $f(0)=0$. Therefore, as $z^t_i\leq 1$ and hence $\gamma z^t_i \leq \gamma$, we have 
    \[
    1-e^{-\gamma z^t_i} \geq \frac{\gamma z^t_i}{\gamma} (1-e^{-\gamma}), \text{which implies } \frac{J^t_i}{\Delta^t_i} \leq \frac{\gamma \theta}{1-e^{-\gamma}}.
    \]
    By substituting $\gamma=3/4$ (and hence $e^{-\gamma}<1/2$) and $\theta=1/(4n)$, we have $J^t_i/\Delta^t_i \leq 3/(8n)$.
\end{proof}

Next we prove $\Phi^t \leq \widehat{\Phi}^t$. Recall that $\Phi^t-\widehat{\Phi}^t=\sum_{i=1}^n J^t_i- \Delta^t_{i_t}$ and Algorithm~\ref{alg:no-informaton} allocates $g_t$ to $i_t$, meaning that $\Delta^t_{i_t} \geq \Delta^t_i$ for all $i$.
By Lemma~\ref{lem::J-and-Delta}, we have
\[
\Phi^t-\widehat{\Phi}^t=\sum_{i=1}^nJ^t_i - \Delta^t_{i_t} \leq \frac{3}{8n}\sum_{i=1}^n\Delta^t_i - \Delta^t_{i_t} \leq -\frac{5}{8}\Delta^t_{i_t} \leq 0.
\]
Therefore, the above inequality together with Lemma~\ref{lem::pre-and-post} directly gives the following.
\begin{proposition}\label{prop::potential-no-info}
    For any round $t\geq 1$, it holds that $\Phi^t \leq \widehat{\Phi}^t \leq \Phi^{t-1}$ and hence $\Phi^t\leq \Phi^0=n$.
\end{proposition}

At this stage, we are ready to use the upper bound of the potential to establish the competitive ratio of Algorithm~\ref{alg:no-informaton}.

\begin{theorem}\label{thm::no-info}
    For any $n\geq 2$, Algorithm~\ref{alg:no-informaton} is deterministic and achieves competitive ratio $\rho=\Omega(\frac{1}{\log (nm)})$. 
\end{theorem}
\begin{proof}%[Proof of Theorem~\ref{thm::no-info}]
    The algorithm is clearly deterministic. We next analyze its performance. 

    For an arbitrary round $t$, we can focus on the agent $i$ with $H^t_i>0$, since agent $i'$ with $H^t_{i'}=0$ is PROP1 trivially. According to Proposition~\ref{prop::potential-no-info}, we have $\Phi^t\leq n$, and hence, $g(\tau^t_i)\exp(\gamma D^t_i/H^t_i) \leq n$.
    Since $g(\tau^t_i)=(\tau^t_i(\tau^t_i+1))^{-1}$, we have $\exp(\gamma D^t_i/H^t_i)\leq n\tau^t_i(\tau^t_i+1) \leq nt(t+1)$, since $\tau^t_i\leq t$ by definition.
    By taking the logarithms and substituting $\gamma=3/4$, we have
    $
    D^t_i \leq \frac{4}{3}\log (nt(t+1)) H^t_i
    $.

    \begin{claim}\label{claim::D^t_i}
        For any agent $i$ and round $t$, $D^t_i=\theta V^t_i - B^t_i$. 
    \end{claim}
    \begin{proof}[Proof of Claim~\ref{claim::D^t_i}]
        We proceed by induction on $t$. For the base case of $t=0$, initialization gives $D^0_i=0 =\theta V^0_i-B^0_i.$
        Suppose that $D^{t-1}_i=\theta V^{t-1}_i - B^{t-1}_i$ holds. Depending on whether agent $i$ receives $g_t$, we have
        $
        V^t_i=V^{t-1}_i+x^t_i $ and $B^t_i=B^{t-1}_i+\mathds{1}_{i_t=i}x_i^t$, where $i_t$ is the agent receiving $g_t$.
        Hence, $D^t_i=D^{t-1}_i+\theta x^t_i - \mathds{1}_{i_t=i}x^t_i$.
        With the induction hypothesis, we have
        \[
        D^t_i = \theta V^{t-1}_i - B^{t-1}_i+\theta x^t_i - \mathds{1}_{i_t=i}x^t_i=\theta(V^{t-1}_i+x^t_i)-(B^{t-1}_i+\mathds{1}_{i_t=i}x^t_i)=\theta V^t_i-B^t_i,
        \]
        completing the proof of the claim.
    \end{proof}
    We also claim that $B^t_i+U^t_i\geq H^t_i$, since if $U^t_i<H^t_i$, then agent $i$ receives the good value $H^t_i$, and thus, $B^t_i \geq H^t_i$.
    Therefore, by $\theta V^t_i=D^t_i+B^t_i$ and the established upper bounds, we have
    \[
    \theta V^t_i \leq B_i^t+\frac{4}{3}\log(nt(t+1))H_i^t \leq \left(1+ \frac{4}{3}\log(nt(t+1))\right) \cdot (B^t_i+U^t_i),
    \]
    which implies
    \[
    \frac{n(B^t_i+U^t_i)}{V^t_i} \geq \frac{n\theta}{1+ \frac{4}{3}\log(nt(t+1))} \geq \frac{1}{4+\frac{16}{3}\log(nt(t+1))},
    \]
    where the last inequality transition is due to $\theta=1/(4n)$.

    Taking $t=m$ proves that desired approximation guarantee. Finally, since $n\geq2$ and $m\geq1$, we have 
    $nm(m+1)\leq (nm)^2$.
Therefore, it holds that $\ln\bigl(nm(m+1)\bigr)
    \leq 2\ln(nm)$,
and hence, the competitive ratio is at least
\[
    \frac{1}{
        4+\frac{32}{3}\log(nm)
    }
    =
    \Omega\left(\frac{1}{\log(nm)}\right),
\]
which completes the proof.
\end{proof}

For any fixed \(n\), the above theorem gives an \(\Omega(\frac{1}{\log m})\) guarantee. The next section shows that the dependence on \(m\) is unavoidable and the best achievable guarantee is within an \(O(\log\log m)\) factor of our established competitive ratio.

\subsection{Without Additional Information: An Impossibility Result}

% Choo et al.~\cite{DBLP:journals/corr/abs-2508-03253} ask what
% multiplicative PROP1 guarantee can be achieved by a deterministic online
% algorithm against an adaptive adversary when no future information is available. The known impossibility of exact PROP1
% rules out competitive ratio one, but does not quantify how the best
% achievable guarantee must depend on the number $m$ of goods or the number of agents $n$. {\color{red}[We have mentioned this in the intro. Perhaps, delete the paragraph.]}

In this section, we complement the result of $\Omega(\frac{1}{\log (nm)})$ algorithm by showing that for every fixed $n\geq 2$, no deterministic online algorithm can achieve better than $O(\frac{n\log\log m}{\log m})$ approximation for PROP1. 
In particular, for any algorithm and any $n\geq 2$, we construct an adaptive instance with only two agents so that the underlying algorithm does not result in competitive ratio better than $\frac{4n\log\log m}{\log m}$.
\begin{theorem}
\label{thm:general-lb}
For any $n\geq 2$, any sufficiently large $m$, and any deterministic online algorithm $ALG$, there exists an adaptive instance
in the setting without additional information with $m$ goods and for which the allocation returned by $ALG$ has PROP1 competitive ratio at
most $\frac{4n\log\log m}{\log m}$.
\end{theorem}

For any round $t$, we say an agent is \emph{active} if she has missed at least one positive valued good.  
For any $i\in[n]$ and round $t$, define
\[
  \rho^t_i=\min\{1, \frac{n(B^t_i+U^t_i)}{V^t_i}\},
\]
which can be viewed as the capped multiplicative PROP1 competitive ratio.
For any agent $i$ and round $t$, define $\rho^t=\min_{i\in N}\rho_i^t$ and let $x^t_i=v_i(g_t)$.
Fix an integer $R\ge2$.  
For any active agent $i$ and round $t$, define
\[
  k_i^t:=\frac{V_i^t-RB_i^t}{U_i^t},
\]
which can be regarded as her \emph{score}.
Rearranging, we get $\frac{k_i^t}{R}=\frac{\frac{V_i^t}{R}-B_i^t}{U_i^t}$, which measures
the number of the  $U_i^t$ valued good we have to owe agent $i$ to make her satisfied at level $\frac{n}{R}$.
Since PROP1 allows agent to add one missed good, if some active agent $i$ at round $t$ has $k_i^t>R$, the allocation has a PROP1 competitive ratio strictly below $\frac{n}{R}$.
We will prove this formally in the next lemma.

\subsubsection*{Normalized Score Dynamics}
\begin{lemma}
  For any round $t$, if there exists an active agent $i$ such that $k_i^t>R$, then $\rho^t<\frac{n}{R}$.
  \label{lem:score-certificate}
\end{lemma}
\begin{proof}
  Since agent $i$ is active, $U_i^t>0$ holds.
  By the definition of $k_i^t$, the inequality $k_i^t>R$ implies $V_i^t-RB_i^t>RU_i^t$.  Thus
  \[
    \frac{n(B_i^t+U_i^t)}{V_i^t}<\frac{n}{R}.
  \]
  Since
  \[
    \rho_i^t = \min\left\{   1,\frac{n(B_i^t+U_i^t)}{V_i^t}\right\}\le   \frac{n(B_i^t+U_i^t)}{V_i^t},
  \]
  and $\rho^t\le\rho_i^t$, we conclude that $\rho^t<\frac{n}{R}$.
\end{proof}
The above lemma gives the connection between $k_i^t$ and $\rho^t$.
Once we get an active agent whose score is strictly larger than $R$, the $\frac{n}{R}$-PROP1 allocation algorithm fails.
So we need transition rules of $k_i^t$ at each round, which will be provided by the next lemma.

\begin{lemma}
  At the beginning of any round $t$, for any active agent $i$ with normalized value $y_i^t=x_i^t/U_i^{t-1}$, the following transition rules hold:
  \begin{enumerate}[label=(\roman*)]
      \item if $i$ receives current good $g_t$, then
    $k_i^t\gets k_i^{t-1}-(R-1)y_i^t$;
      \item if $i$ misses $g_t$ and $0\le y_i^t\le1$, then
    $k_i^t\gets k_i^{t-1}+y_i^t$;
      \item if $i$ misses $g_t$ and $1<y_i^t$, then
    $k_i^t\gets 1+k_i^{t-1}/y_i^t$.
  \end{enumerate}
  \label{lem:exact-score-transitions}
\end{lemma}
\begin{proof}
  At the beginning of any round $t$ and for any agent $i$, $V_i^t=V_i^{t-1}+y_i^tU_i^{t-1}$ holds.
  If $i$ receives $g_t$, then $B_i^t=B_i^{t-1}+y_i^tU_i^{t-1}$.
  By the definition of $k_i^t$, we have
  \[
    k_i^t=\frac{V_i^{t-1}+y_i^tU_i^{t-1}-RB_i^{t-1}
      -Ry_i^tU_i^{t-1}}{U_i^{t-1}}=k_i^{t-1}-(R-1)y_i^t.
  \]
  Now suppose $i$ misses $g_t$, then $B_i^t=B_i^{t-1}$ and consider two cases.
  If $0\le y_i^t\le1$, then $x_i^t\le U_i^{t-1}$.  Thus $U_i^t=U_i^{t-1}$.
  Therefore
  \[
    k_i^t=\frac{V_i^{t-1}+y_i^tU_i^{t-1}-RB_i^{t-1}}{U_i^{t-1}}=k_i^{t-1}+y_i^t.
  \]
  If $y_i^t>1$, then $x_i^t>U_i^{t-1}$.  Thus $U_i^t=y_i^tU_i^{t-1}$.  Consequently
  \[
    k_i^t=\frac{V_i^{t-1}+y_i^tU_i^{t-1}-RB_i^{t-1}}{y_i^tU_i^{t-1}}
    =1+\frac{k_i^{t-1}}{y_i^t}.
  \]
  We proved the three transition rules.
\end{proof}

The three branches of Lemma~\ref{lem:exact-score-transitions} play distinct roles in transition.
The first rule is a direct way to lower a score: once one agent receives one positive valued good, her score will decrease.
The second rule is a direct way to raise a score: if an agent $i$ misses one positive valued good with normalized value no more than $1$, her score moves upward.
This will be the core transition to make one agent's score exceed $R$.
The third one is subtle, for it may increase or decrease a score.
But it has a unique feature.
If agent $i$ misses a good with a sufficiently large normalized value, then her score lies in an arbitrarily small neighborhood of $1$, no matter how negative it was.

\subsubsection*{Length-Penalized Forcing Construction}
Before presenting the formal proof of Theorem~\ref{thm:general-lb}, we first explain the idea of our construction.
With Lemma~\ref{lem:score-certificate}, for any round $t$, once there exists one agent $i$ that satisfies $k_i^t>R$, the PROP1 competitive ratio is less than $\frac{n}{R}$.
The goal of the adversary is therefore to force the $k_i^t$ of some agent $i$ larger than $R$, at round $t$.

Our construction is for two agents and can be extended to any $n$ agents by adding $n-2$ dummy agents.
The constructed instance can be divided into four phases.
In Phase $1$, there are at most $R+1$ goods, each having value $1$ for both agents and value $0$ for all other dummy agents.
If none of their scores exceeds $R$, then the two agents will be active within the $R+1$ goods.
Then in Phase $2$,  within a bounded number of goods, the adversary drives agent $1$ very close to break her supremum, which we will discuss further later.
In Phase $3$, the adversary uses one good to confine agent $2$'s score to $(0, 2)$.
Finally in Phase $4$, the adversary constructs several goods, and their values are designed in a way that each of them must be allocated to agent $1$; otherwise her supremum guarantee fails.
However, agent $2$'s score will exceed $R$.
After this construction we append zero valued goods until the total number of goods is exactly $m$.

\begin{figure}[htbp]
  \centering
  \includegraphics[width=0.9\linewidth]{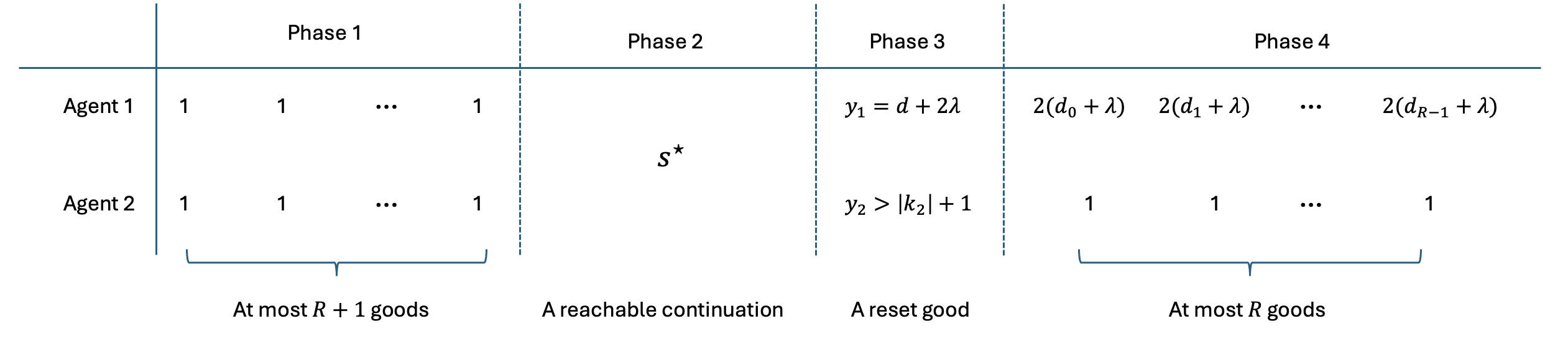}
  \caption{Illustration of the counterexample construction in the proof of Theorem~\ref{thm:general-lb}.}
  \label{fig:length-penalized-forcing-construction}
\end{figure}

\begin{proof}[Proof of Theorem~\ref{thm:general-lb}]
  For a contradiction, assume that there exists a deterministic algorithm $ALG$ such that, at any round $t$ the score of every active agent never exceed $R$.

In Phase $1$, there are at most $R+1$ goods with value $1$ for agent $1$ and $2$, and value $0$ for all other dummy agents.
If $ALG$ allocates the first good to a dummy agent, then both are active, and the adversary construction enters Phase $2$.
Otherwise, without loss of generality, we assume the first good is allocated to agent $1$.
If the remaining $R$ goods are all allocated to agent $1$, then agent $2$ is active with $k_2^{R+1}=R+1>R$.
That contradicts the hypothesis that for any active agent $i$ and any round $t$, $ALG$ maintains $k_i^t\le R$.
So at least one of the first $R+1$ is not allocated to agent $1$.
Let $T\le R+1$ be the first such round.
The good $g_T$ makes agent $1$ active and agent $2$ is active already, so both agent $1$ and agent $2$ are active at round $T$.

Because both agents have missed at least one good, for any $i=1,2$, $B_i^T\le T-1$, $U_i^T=1$ and $V_i^T=T$ hold.
Therefore,
\[
  k_i^T = V_i^T-RB_i^T \ge T-R(T-1) \ge R-(R-1)(R+1)=1+R-R^2,
\]
where the last inequality transition follows from $T\le R+1$.

In Phase $2$, for convenience, we restart the round index at zero, that is, for any $i=1,2$, $1+R-R^2\le k_i^0\le R$ holds, where $k_i^0$ denotes the score right after Phase $1$.
Define $\Lambda=2R-1$, and
\[
  \lambda=\frac{1}{4(3R+2)\Lambda^R}, \quad
  T_0=\left\lceil\frac{R^2}{\lambda}\right\rceil+R.
\]

Let $\mathcal{A}$ be a deterministic algorithm, and let $H_t$ denote the history through round $t$, that is, the goods revealed in round $1,\ldots,t$, together with the recipients $\mathcal{A}$ assigned to them.
A \emph{continuation} from round $t$ is a finite sequence
\[
  s=\bigl((x^1,a_1),\ldots,(x^\ell,a_\ell)\bigr)
\]
of value vectors with their recipients, describing the $\ell$ consecutive rounds following round $t$, and its \emph{length} $\ell(s)=\ell$ is the number of its goods.
For $1\le j\le\ell+1$, write
\[
  s_{<j}=\bigl((x^1,a_1),\ldots,(x^{j-1},a_{j-1})\bigr)
\]
for the round of $s$. Write $H_ts$ for the concatatenation of $H_t$ and $s$, which means the history of $t+\ell$ rounds obtained.
The continuation is \emph{reachable} if for any $j\in [\ell]$
\[
  a_j=\mathcal A\bigl(H_ts_{<j},\,x^j\bigr)
\]
that is, any recipient in $s$ is chosen by $\mathcal{A}$.

For $ALG$, let $\mathcal{S}$ denotes the set of all reachable continuations from round zero reached at the end of Phase $1$, with length at most $T_0$.
For any $s\in\mathcal{S}$ and any active agent $i$, we write $\ell(s)$ for the continuation's length and $k_i(s)$ for agent $i$'s score at its endpoint.
By the hypothesis on $ALG$, no reachable continuation in $\mathcal{S}$ drives either agent's score above $R$.
Equivalently, for any $s\in\mathcal S$ and any $i\in[2]$, $k_i(s)\le R$ holds.

For any $s\in \mathcal{S}$, define agent $1$'s \emph{length-penalized score} $F(s)=k_1(s)-\lambda\ell(s)$.
Let
\[
  M=\sup_{s\in\mathcal{S}}\{F(s)\}=
  \sup_{s\in\mathcal{S}}\{k_1(s)-\lambda\ell(s)\}.
\]
Because the empty continuation belongs to $\mathcal{S}$, $M\ge k_1^0\ge 1+R-R^2$.
By the hypothesis, we have $M \le R$.
With the definition of $M$, there exists a reachable continuation $s^\star\in\mathcal{S}$ such that $F(s^\star)>M-\lambda$.
Let $d=M-(k_1(s^\star)-\lambda\ell(s^\star))<\lambda$.
Thus
\[
  \lambda\ell(s^\star) < k_1(s^\star)+\lambda-M < R+\lambda-(1+R-R^2) < R^2.
\]
So $\ell(s^\star)\le\left\lceil R^2/\lambda\right\rceil-1$.

In Phase $3$, we will use one good to confine agent $2$'s score to $(0, 2)$.
Reveal a good with normalized value $y_1=d+2\lambda$ and $y_2>|k_2|+1$.
Since $d<\lambda$, it follows that $y_1<3\lambda<1$.
If agent $1$ misses it, her length-penalized score is
\[
  k_1(s^\star)+y_1-\lambda(\ell(s^\star)+1) = M-d-\lambda+y_1 = M+\lambda > M,
\]
contradicting the definition of $M$.
Therefore this good must be allocated to agent $1$.
Let $k_{1,0}$ be agent $1$'s score and $k_{2,0}$ be agent $2$'s score and $\ell_0$ be the new continuation length after this allocation, i.e., $\ell_0=\ell(s^\star)+1$.
By Lemma~\ref{lem:exact-score-transitions} $k_{1,0}=k_1-(R-1)(d+2\lambda)$ and $k_{2,0}=1+\frac{k_2}{y_2}$.  Since $|\frac{k_2}{y_2}|<1$, it follows that $0<k_{2,0}<2$.
Define
\begin{equation}
 d_0=M-[k_1(s^\star)-(R-1)(d+2\lambda)-\lambda\ell_0]=d+(R-1)(d+2\lambda)+\lambda< (3R-1)\lambda.
  \label{eq:d_0-bound}
\end{equation}

This good yields a score for agent $2$ that is greater than zero and less than two.
A receipt cannot increase an agent's score, whereas a miss of normalized value one increases it by exactly one.
We therfore set the normalized value of every good in the next phase to one for agent $2$, and try to force the allocation of the remaining goods to agent $1$.
If we succeed in doing so for $R$ consecutive rounds, then agent $2$'s scire will exceed $R$.

In Phase $4$, we number the goods of this phase $1,\ldots,R$.
For any round $t\in[R]$, let $k_{i,t}$ denote agent $i$'s score and let $\ell_t$ be the corresponding continuation length after the first $t$ goods of this phase have been allocated, so that $\ell_t=\ell_0+t$.
In particular, $t=0$ refers to the state immediately after Phase $3$, and $k_{i,0}$, $\ell_0$ are inherited from that phase.

For any round $t\in[R]$ define $d_t=M-(k_{1,t}-\lambda\ell_t)$.
Reveal a good $g_t$ with normalized value
$y_{1,t}=2(d_{t-1}+\lambda)$ and $y_{2,t}=1$.
Set $c=\Lambda/(\Lambda-1)$.
We will prove that for any round $t\ge1$, $0\le y_{1,t}\le1$, $d_t+c\lambda=\Lambda^t(d_0+c\lambda)$ and $g_t$ will be allocated to agent $1$, with induction on $t$.

\emph{Base case}: In the first round $y_{1,1}=2(d_0+\lambda)\ge0$ and with $d_0<(3R+2)\lambda$ we have
\[
  y_{1,1}<2[(3R+2)\lambda+\lambda]<1.
\]
If $g_1$ is  not allocated to agent $1$, then $k_{1,1}=k_{1,0}+y_{1,1}=M+d_0+\lambda(\ell+3)$.
Thus her length-penalized score will be $k_{1,1}-\lambda\ell_1=M+d_0+\lambda>M$, contradicting the definition of $M$.
Therefore $ALG$ must allocate it to agent $1$.
Then $k_{1,1}=k_{1,0}-(R-1)y_{1,1}$, so $d_1=(2R-1)(d_0+\lambda)$.
Rearranging, we get $d_1+c\lambda=\Lambda(d_0+c\lambda)$.

\emph{Induction hypothesis}: Fix $r$ with $2\le r\le R$, and assume $0\le y_{1,r-1}\le1$, $d_{r-1}+c\lambda=\Lambda^{r-1}(d_0+c\lambda)$ and $g_r$ is allocated to agent $1$.

\emph{Induction step}: For Inequality (\ref{eq:d_0-bound}) and $c<2$,
\[
  d_{r-1}+c\lambda = \Lambda^{r-1}(d_0+c\lambda)< \Lambda^{r-1}(3R+1)\lambda.
\]
So
\[
  y_{1,r} = 2(d_{r-1}+\lambda)< 2\lambda[\Lambda^{r-1}(3R+1)+1]<2\lambda\Lambda^{r-1}(3R+2)=\frac{\Lambda^{r-1}}{2\Lambda^R}< 1.
\]
If $g_r$ is not allocated to agent $1$, then $k_{1,r}=k_{1,r-1}+y_{1,r}$, so
\[
  d_r = M-[k_{1,r-1}+2(d_{r-1}+\lambda)-\lambda\ell_{r-1}-\lambda] = -d_{r-1}-\lambda < 0,
\]
contradicting the definition of $M$.
Then $g_r$ is allocated to agent $1$, and $k_{1,r}=k_{1,r-1}-(R-1)y_{1,r}$.
Thus
\[
  d_r = M-(k_{1,r}-\lambda\ell_r) = d_{r-1}+\lambda+(R-1)2(d_{r-1}+\lambda) = (2R-1)(d_{r-1}+\lambda).    
\]
Rearranging, we get $d_r+c\lambda=\Lambda(d_{r-1}+c\lambda)$. Thus
$d_r+c\lambda=\Lambda^r(d_0+c\lambda)$.

The induction proved that all these goods will be allocated to agent $1$.
Thus $k_{2,t}=k_{2,t-1}+1$, and hence, by iteration for any $t\in[R]$, $k_{2,t}=k_{2,0}+t$.
Since Phase $3$ ensures that $k_{2,0}>0$, after the $R$ rounds, we obtain $k_{2,R}=k_{2,0}+R>R$.
Therefore agent $2$'s score exceeds $R$ during this phase, and the continuation built in Phase $2$ to $4$ is reachable and has length at most $T_0$, so it lies in $\mathcal{S}$, contradicting our hypothesis.
We conclude that for any integer $R\ge2$ and any deterministic algorithm, the adversary is able to force the score of some agent above $R$, and by Lemma~\ref{lem:score-certificate} the PROP1 factor at this construction is below $\frac{n}{R}$.

Phase $1$ uses at most $R+1$ goods, Phase $2$ to $4$ at most $T_0$ goods.
Hence the total length of the construction is no greater than $R+1+T_0=\left\lceil\frac{R^2}{\lambda}\right\rceil+2R+1$.
For $R\ge2$, we have $3R+2\le4R$ and $2R-1\le2R$.  Thus
\[
  \frac{R^2}{\lambda}=4R^2(3R+2)(2R-1)^R\le16R^3(2R)^R.
\]
Therefore the total length of the construction is at most $\left\lceil\frac{R^2}{\lambda}\right\rceil+2R+1\le32R^3(2R)^R\le(2R)^{R+5}$.
Fix an arbitrary sufficiently large integer $m$.  Set
\[
  R=\left\lfloor\frac{\log m}{2\log\log m}\right\rfloor.  
\]
Because $m$ is sufficiently large, the quantity inside the floor is at least $4$.  Hence
\[
  R\ge\frac{\log m}{4\log\log m}.
\]
Define $L=\log m$ and $\ell_m=\log\log m=\log L$.
Because $m$ is sufficiently large, $L\ge10\ell_m$.
Thus $5\le L/(2\ell_m)$.
Since $R\le L/(2\ell_m)$
\[
  \log((2R)^{R+5})=(R+5)\log(2R)\le(\frac{L}{2\ell_m}+\frac{L}{2\ell_m})\log\frac{L}{\ell_m}=\frac{L}{\ell_m}\log\frac{L}{\ell_m}.
\]
For sufficiently large $m$, $\ell_m\ge1$, so $L/\ell_m\le L$.  Thus $\log((2R)^{R+5})\le L=\log m$.

Running the adversary construction until the first round at which some agent's score exceeds $R$.
The preceding length estimate shows that this occurs within at most $m$ goods.
If the violation occurs before round $m$, append zero-valued goods until the total number is $m$.
This padding will not change $V_i$, $B_i$, $U_i$, or the PROP1 competitive ratio.
Let agent $i$ be the first whose score exceeds $R$ in round $t$.
By Lemma \ref{lem:score-certificate}, $\rho^t<\frac{n}{R}$.
It follows that
\[
  \rho\le\frac{4n\log\log m}{\log m},
\]
for the sufficiently large $m$.
\end{proof}

\section{With Accurate Maximum Item Value Predictions}\label{sec::accurate-MIV}

In this section, we are concerned with the setting when the maximum item value predictions are accurate.
We present a deterministic algorithm achieving competitive ratio 0.5, and complement it with an upper bound strictly less than one. 

\subsection{With Accurate MIV Predictions: A Deterministic Algorithm}

In the setting without additional information, the algorithm uses
$H_i^t$, the largest value observed for agent $i$ by round $t$, as a
normalization value. Since $H_i^t$ may increase as new goods arrive, the
potential must account for changes in this value. With accurate MIV
information, the value $p_i$ is known before the first round and remains
fixed throughout the allocation process. We use this fixed value to
normalize agent $i$'s valuations and construct a reciprocal potential.
For every agent $i\in N$, recall that $p_i$ is her MIV, which is known before the first round.
Define $N^+=\{i\in N:p_i>0\}$ and $\beta=\frac{1}{2n}$.
For every $i\notin N^+$, we have $v_i(g_t)=0$ in every round $t$. Such agents satisfy PROP1 trivially.
We therefore define the quantities used by the algorithm only for agents
in $N^+$.

We first informally introduce the algorithm. For any $i$ and round $t$, right after allocating $g_t$, the algorithm maintains three quantities $\delta^t_i$, $h(\textsc{Sta}^t_i)$, and $q^t_i$.
The $\delta^t_i$ represents the normalized difference between agent $i$'s value and $1/2$ fraction of her proportional share right after allocating $g_t$. 
Once $g_t$ has been allocated, the algorithm updates $\delta^t_i$ as follows: if $i$ receives $g_t$, her $\delta$ decreases by $1-\beta$, and otherwise, increases by $\beta$.
This update rule ensures that $\delta^t_i=\frac{V^t_i-B^t_i}{p_i}$, which we will formally prove below. 
The $h(\textsc{Sta}^t_i)$ depends on \emph{status} $\textsc{Sta}^t_i$ and represents the correction term from the goods outside the bundle of agent $i$.
We will explain status and define the corresponding $h$ values below. Informally, if agent $i$ misses some good with value 1 for her, then the $h$ value is set to one. If agent $i$ receives every good with value 1 for her, then the $h$ value is set to zero.
The $q^t_i$ is equal to $h(\textsc{Sta}^t_i) - \delta^t_i$, which can be viewed as the safety distance of agent $i$. At termination, we will prove that for any agent $i$, $q^t_i>0$, which then gives the desired approximation factor.
Moreover, we treat $\frac{1}{q^t_i}$ as agent $i$'s contribution to the \emph{potential}, defined as $\Psi^t=\sum_i \frac{1}{q^t_i}$. We choose the reciprocal potential, as it assigns enormous urgency to an agent whose safety distance is close to zero.

Since $p_1,\ldots,p_n$ are known before the arrival for the first good, when $g_t$ arrives, we observe $x^t_1,\ldots,x^t_n$ and compute $y^t_i=x^t_i/p_i$ the normalized value of $g_t$ for agent $i$.
Then the algorithm classifies the agents into two classes, $\textsc{Spe}^t$ and $\textsc{Ord}^t$, where each agent in $\textsc{Spe}^t$ is called special and each agent in $\textsc{Ord}^t$is called ordinary.
An agent $i$ is in $\textsc{Spe}^t$ if $y^t_i=1$ and no good with value 1 for $i$ has been allocated to other agent before the allocation of $g_t$. Otherwise, the agent $i$ is in $\textsc{Ord}^t$.
The algorithm uses $G^t_i$'s to decide the recipient of $g_t$. For each ordinary agent $i$, $G^t_i$ is the reduction in the reciprocal of safety distance obtained by assigning $g_t$ to $i$, compared with letting agent $i$ miss $g_t$.
And for each special agent $i$, set $G^t_i=0$. The algorithm then allocates $g_t$ to the agent having the highest $G^t_i$.
The formal description is provided in Algorithm~\ref{alg:MIV-0.5}.

We now formally introduce status and its update rule. For any agent $i$ and round $t$, the algorithm maintains  $\textsc{Sta}^t_i$ right after the allocation of $g_t$. 
There are in total three status, namely, $\textsc{Sta1},\textsc{Sta2},\textsc{Sta3}$.
% In other words, $\textsc{Sta}^t_i \in \{\textsc{Sta1}, \textsc{Sta2}, \textsc{Sta3}\}$. 
Status tell whether the good with normalized value 1 has appeared for the agent, and whether at least one good with normalized value 1 has been assigned to another agent.
For better expositions, we explain their informal representations below;
\begin{itemize}
    \item $\textsc{Sta}^t_i=\textsc{Sta1}$ informally represents that no good with normalized value 1 for $i$ has appeared at round $t$;
    \item $\textsc{Sta}^t_i=\textsc{Sta2}$ informally represents that at least one good with normalized value 1 for $i$ has appeared, and every such good has been assigned to agent $i$;
    \item $\textsc{Sta}^t_i=\textsc{Sta3}$ informally represents that at least one good with normalized value 1 for $i$ has been assigned to some agent $j\neq i$.
\end{itemize}

Initially, for each $i\in N^+$, $\textsc{Sta}^0_i=\textsc{Sta1}$. For any round $t$, after $g_t$ is allocated, the status of agent $i$ is updated based on the following rule.
We let $i_t$ denote the agent receiving $g_t$ in the algorithm. For any $i$ and round $t$, 
\begin{itemize}
    \item if $y^t_i<1$, then $\textsc{Sta}^t_i\gets \textsc{Sta}^{t-1}_i$;
    \item if $y^t_i=1$, $\textsc{Sta}^{t-1}_i = \textsc{Sta1}$ and $i_t=i$, then $\textsc{Sta}^t_i\gets \textsc{Sta2}$;
    \item if $y^t_i=1$, $\textsc{Sta}^{t-1}_i = \textsc{Sta1}$ and $i_t\neq i$, then $\textsc{Sta}^t_i\gets \textsc{Sta3}$;
    \item if $y^t_i=1$, $\textsc{Sta}^{t-1}_i = \textsc{Sta2}$ and $i_t = i$, then $\textsc{Sta}^t_i\gets \textsc{Sta2}$;
    \item if $y^t_i=1$, $\textsc{Sta}^{t-1}_i = \textsc{Sta2}$ and $i_t \neq i$, then $\textsc{Sta}^t_i\gets \textsc{Sta3}$;
     \item if $y^t_i=1$ and $\textsc{Sta}^{t-1}_i = \textsc{Sta3}$, then $\textsc{Sta}^t_i\gets \textsc{Sta3}$;
\end{itemize}
One can verify that the described intuitions of status is compatible with the update rule. Define $h(\textsc{Sta}_1)=1-\beta$, $h(\textsc{Sta}_2)=0$, and $h(\textsc{Sta}_3)=1$. Since $h$ refers to the correction term outside the bundle of the agent, the values of $h(\textsc{Sta}_2)$ and $h(\textsc{Sta}_3)$ are quite clear. As for $h(\textsc{Sta}_1)$, we make its value as $1-\beta$ so that the first good with value $p_i$ for $i$ does not affect $q_i$, regardless of who receives it.

\begin{algorithm}[t]
	\caption{Deterministic allocation with exact MIV information}
	\label{alg:MIV-0.5}
 \renewcommand{\algorithmicensure}{\textbf{Output:}}
	\begin{algorithmic}[1]
        \REQUIRE A fixed deterministic ordering of the agents for tie-breaking.
        \STATE Initialize $\delta^0_i=0, \textsc{Sta}^0_i=\textsc{Sta}_1, q^0_i=1-\beta$ and the empty allocation $A=(A_1,\ldots,A_n)$ with $A_i=\varnothing$ for all $i$.
        \WHILE{there exists a good $g_t$ arriving online}
        \STATE Observe $x^t=(x^t_1,\ldots,x^t_n)$ and compute $y^t_i=x^t_i/p_i$ for every $i\in N^+$.
        \STATE Define $\textsc{Spe}^t=\{i\in N^+:y^t_i=1 \text{ and } \textsc{Sta}^{t-1}_i\in \{\textsc{Sta1},\textsc{Sta2}\}\}$ and $\textsc{Ord}^t=N^+\setminus \textsc{Spe}^t$.
        \STATE For each $i\in \textsc{Ord}^t$, define $
        G^t_i=1/q^{t,-}_i - 1/q^{t,+}_i
        $, where $q^{t,-}_i=q^{t-1}_i-\beta y^t_i$ and $q^{t,+}_i=q^{t-1}_i+(1-\beta) y^t_i$.
        \STATE For each $i\in \textsc{Spe}^t$, define $G^t_i=0$.
        \STATE Allocate $g_t$ to agent $i_t \in \arg\max_{i\in [n]} G ^t_i$, i.e., $A_{i_t}\gets A_{i_t}\cup \{g_t\}$. Break ties according to the fixed deterministic ordering.
        \STATE Update: (1) for every $i\in N^+$, $\delta^t_i\gets \delta^{t-1}_i+\beta y^t_i-\mathds{1}_{\{i_t=i\}} y^t_i$, (2) update $\textsc{Sta}^t_i$ based on the described update rules, and (3) for every $i\in N^+$, $q^t_i=h(\textsc{Sta}^t_i)-\delta^t_i$.\label{step::alg-MIV-update}
        \ENDWHILE
        \ENSURE Allocation $A$.
	\end{algorithmic}
\end{algorithm}

Let us now prove the competitive ratio. We begin with the updated values of $\delta^t_i$'s and $q^t_i$'s after the allocation of $g_t$.

\begin{proposition}\label{prop::MIV-detla}
    For any $i\in N^+$ and round $t$, $\delta^t_i=\frac{\beta V^t_i - B^t_i}{p_i}$.
\end{proposition}
\begin{proof}
    We prove by induction on $t$. For base case $t=0$, the equation holds by definition of $\delta^t_i$, $V^t_i$, and $B^t_i$. Suppose that the statement holds at round $t-1$. Recall that $i_t$ is the recipient of $g_t$. Then we have
    \[
    \delta^t_i = \beta \frac{V^{t-1}_i}{p_i} - \frac{B^{t-1}_i}{p_i} + \beta y^t_i - \mathds{1}_{\{i_t=i\}}y^t_i = \beta\frac{V^{t-1}_i+x^t_i}{p_i} - \frac{B^{t-1}_i+\mathds{1}_{\{i_t=i\}}x^t_i}{p_i} = \frac{\beta V^t - B^t_i }{p_i}.
    \]
    This completes the induction.
\end{proof}

\begin{proposition}\label{prop::MIV-q}
    For agent $i\in N^+$ and round $t$, 
    \begin{enumerate}[label=(\roman*)]
        \item if $i \in \textsc{Ord}^t$ and $i$ receives $g_t$, $q^t_i=q^{t-1}_i+(1-\beta)y^t_i$; if $i \in \textsc{Ord}^t$ and $i$ misses $g_t$, $q_i^t=q^{t-1}_i-\beta y^t_i$.
        \item if $i \in \textsc{Spe}^t$ and $\textsc{Sta}^{t-1}_i=\textsc{Sta1}$, then $q^t_i=q^{t-1}_i$; if $i \in \textsc{Spe}^t$ and $\textsc{Sta}^{t-1}_i=\textsc{Sta2}$, then $q^t_i=q^{t-1}_i+1-\beta$. 
    \end{enumerate}
\end{proposition}
\begin{proof}
    We first prove (i). Based on the status update rule, when $i\in \textsc{Ord}^t$, we have $\textsc{Sta}^{t-1}_i = \textsc{Sta}^{t}_i$. Thus, no matter whether agent $i$ receives $g_t$, $q^t_i = q^{t-1}_i - \delta^t_i +\delta^{t-1}_i$.
    If agent $i$ receives $g_t$, we have $\delta^t_i-\delta^{t-1}_i = (\beta-1) y^t_i$ due to Proposition~\ref{prop::MIV-detla}. Hence, in this case, $q^t_i=q^{t-1}_i+(1-\beta)y^t_i$.
    Again, by Proposition~\ref{prop::MIV-detla}, if agent $i$ does not receive $g_t$, $\delta^t_i - \delta^{t-1}_i = \beta y^t_i$, which implies that $q_i^t=q^{t-1}_i-\beta y^t_i$.

    Next we prove (ii). For the case when $i \in \textsc{Spe}^t$ and $\textsc{Sta}^{t-1}_i=\textsc{Sta1}$, if $i$ receives $g_t$, then $\textsc{Sta}^{t}_i=\textsc{Sta2}$.
    Then by Proposition~\ref{prop::MIV-detla}, we have $q^t_i-q^{t-1}_i=h(\textsc{Sta2}) - h( \textsc{Sta1})- (\delta^t_i-\delta^{t-1}_i)=0$.
    Similarly, if $i$ does not receive $g_t$, we have $\textsc{Sta}^{t}_i=\textsc{Sta3}$. Then it is not hard to verify $q^t_i=q^{t-1}_i$.
    Finally, for the case when $i \in \textsc{Spe}^t$ and $\textsc{Sta}^{t-1}_i=\textsc{Sta2}$, if $i$ receives $g_t$, then $\textsc{Sta}^{t}_i=\textsc{Sta2}$.
    By Proposition~\ref{prop::MIV-detla}, we have $\delta^t_i-\delta^{t-1}_i=\beta -1$, and thus, $q^t_i=q^{t-1}_i + 1-\beta$.
    If $i$ misses $g_t$, then $\textsc{Sta}^{t}_i=\textsc{Sta3}$ and $\delta^t_i-\delta^{t-1}_i=\beta$. Consequently, $q^t_i=q^{t-1}_i + 1-\beta$.
\end{proof}

The above proposition indicates if agent $i$ is special at the beginning of round $t$, her $q_i$ does not decrease, regardless of the recipient $g_t$. Next, we prove the crucial lemmas for establishing the main result. Formally, define $\Psi^t=\sum_{i\in N^+} 1/q^t_i$ the potential at round $t$. In particular, $\Psi^0=n/(1-\beta)$.
\begin{lemma}
\label{lemma::MIV-implication}
Fix a round $t\geq1$. Suppose that
$q_i^{t-1}>0$ for every $i\in N^+$ and
$
    \Psi^{t-1}\leq\frac{n}{1-\beta}.
$
Then $q_i^t>0$ for every $i\in N^+$ and
$
    \Psi^t\leq\Psi^{t-1}.
$
\end{lemma}

\begin{proof}
We first prove positivity. Since
$1/q_i^{t-1}\leq\Psi^{t-1}$, for every $i\in N^+$, we have
$
    q_i^{t-1}\geq\frac{1-\beta}{n}.
$
For an ordinary agent,
\[
\begin{aligned}
    q_i^{t,-}
    =q_i^{t-1}-\beta y_i^t\geq \frac{1-\beta}{n}-\beta=\frac{n-1}{2n^2}>0,
\end{aligned}
\]
where we use $y_i^t\leq1$ and $\beta=1/(2n)$. Also,
$q_i^{t,+}>0$. For a special agent,
Proposition~\ref{prop::MIV-q} gives
$q_i^t\geq q_i^{t-1}>0$. Therefore, every realized value $q_i^t$
is positive.

We next prove that the potential does not increase. Define
\[
    \overline{\Psi}^t
    =
    \sum_{i\in\textsc{Ord}^t}\frac{1}{q_i^{t-1}}
    +
    \sum_{i\in\textsc{Spe}^t}\frac{1}{q_i^t}.
\]
For every special agent, $q_i^t\geq q_i^{t-1}$, and hence
$
    \overline{\Psi}^t\leq\Psi^{t-1}.
$

For every $i\in\textsc{Ord}^t$, define
$
    F_i^t
    :=
    \frac{1}{q_i^{t,-}}-\frac{1}{q_i^{t-1}}.
$
Let
$
    G_*^t:=\max_{i\in N}G_i^t=G_{i_t}^t.
$
If $G_*^t=0$, then $G_i^t=0$ for every ordinary agent, which implies
$y_i^t=0$ and hence $F_i^t=0$. The ordinary contributions do not
increase, while the special contributions do not increase by
Proposition~\ref{prop::MIV-q}. Thus,
$\Psi^t\leq\overline{\Psi}^t$.

Suppose now that $G_*^t>0$. Then the selected recipient $i_t$ is
ordinary, and
\[
    \Psi^t-\overline{\Psi}^t
    =
    \sum_{i\in\textsc{Ord}^t}F_i^t-G_*^t.
\]
For every ordinary agent with $y_i^t>0$, direct calculation gives
\[
\begin{aligned}
    \frac{F_i^t}{G_i^t}
    =
    \beta
    \frac{q_i^{t-1}+(1-\beta)y_i^t}{q_i^{t-1}}=
    \beta+
    \beta(1-\beta)\frac{y_i^t}{q_i^{t-1}}.
\end{aligned}
\]
The same identity holds when $y_i^t=0$, since then
$F_i^t=G_i^t=0$. Therefore,
\[
\begin{aligned}
\sum_{i\in\textsc{Ord}^t}F_i^t
&=
\beta\sum_{i\in\textsc{Ord}^t}G_i^t
+
\beta(1-\beta)
\sum_{i\in\textsc{Ord}^t}
\frac{y_i^t}{q_i^{t-1}}G_i^t\\
&\leq
\beta nG_*^t
+
\beta(1-\beta)G_*^t
\sum_{i\in\textsc{Ord}^t}\frac{1}{q_i^{t-1}}\\
&\leq
\beta nG_*^t+
\beta(1-\beta)\Psi^{t-1}G_*^t\\
&\leq
2\beta nG_*^t
=
G_*^t.
\end{aligned}
\]
It follows that
$
    \Psi^t\leq\overline{\Psi}^t\leq\Psi^{t-1}.
$
\end{proof}

By the above lemma, we can derive the following properties.

\begin{lemma}\label{lemma::MIV-property}
    For any agent $i\in N^+$ and round $t\geq 1$, it holds that (i) $q^t_i>0$ and (ii) $\Psi^t \leq \Psi^{t-1}$.
\end{lemma}
\begin{proof}
    We prove by induction. For the base case of $t=1$, by Lemma~\ref{lemma::MIV-implication} and facts that $q^0_i>0$ and $\Psi^0=n/(1-\beta)$, it holds that $q^1_i>0$ and $\Psi^1 \leq \Psi^0$. For the induction step, suppose that the two properties hold for every round $1,\dots, t-1$. Accordingly, $q^{t-1}_i>0$. Moreover, by repeatedly applying property (2), we have
    \[
    \Psi^{t-1} \leq \Psi^{t-2} \leq \dots \leq \Psi^1 \leq \Psi^0=\frac{n}{1-\beta}.
    \]
    Again, by Lemma~\ref{lemma::MIV-implication}, we have $q^t_i>0$ and $\Psi^{t} \leq \Psi^{t-1}$.
    Finally, by induction, the lemma statement holds.
\end{proof}

We now use the positivity established in Lemma~\ref{lemma::MIV-property} to prove the main result.
\begin{theorem}
    For any $n\geq 2$, Algorithm~\ref{alg:MIV-0.5} is deterministic and achieves competitive ratio 1/2.
\end{theorem}
\begin{proof}
    The algorithm is clearly deterministic. We next analyze its performance. Fix an arbitrary round $t$. If $i \notin N^+$, we have $v_i(g_j)=0$ for all $j\leq t$. Thus, agent $i$ satisfies PROP1. In the remaining, we focus on $ i \in N^+$. 
    
    As $\beta=1/(2n)$, proving the 1/2 competitive ratio is reduced to prove $B^t_i + U^t_i \geq \beta V^t_i$. By Proposition~\ref{prop::MIV-detla}, it holds that $\beta V^t_i - B^t_i = p_i\delta ^t_i $. Moreover, as $\delta^t_i= h(\textsc{Sta}^t_i) - q^t_i $ (by Line~\ref{step::alg-MIV-update}), we have $\beta V^t_i - B^t_i = p_i (h(\textsc{Sta}^t_i) - q^t_i)$.
    We note that when round $t$ is the termination, the good with predicted value $p_i$ must come. Thus, $\textsc{Sta}^t_i \in \{ \textsc{Sta2}, \textsc{Sta3} \}$. 
    If $\textsc{Sta}^t_i = \textsc{Sta2}$, then we have $\beta V^t_i - B^t_i = -p_iq^t_i < 0$, and hence, $B^t_i+U^t_i>\beta V^t_i$.
    If $\textsc{Sta}^t_i = \textsc{Sta3}$, then we have $\beta V^t_i - B^t_i = p_i (1-q^t_i) < p_i = U_i$ because $q^t_i>0$. Thus, $B^t_i+U^t_i>\beta V^t_i$.

    Therefore, we have $B^t_i+U^t_i> V^t_i/(2n)$, completing the proof.
\end{proof}
We remark that as shown in the above proof, the potential bound $\Psi^t \leq \Psi^0$ does not directly establishes the competitive ratio. Instead, Lemma~\ref{lemma::MIV-property} gives $\Psi^{t-1} \leq \Psi^0$ and then $q^t_i>0$, based on which we are able to establish $B^t_i+U^t_i \geq \beta V^t_i$.

\subsection{With Accurate MIV Predictions: An impossibility result}

The preceding algorithm gives a constant guarantee under exact MIV information. 
We now show that this information does not suffice to make the competitive ratio arbitrarily close to one, even for two agents. The construction fixes the maximum item values at \(p_1=p_2=1\) and uses at most \(518\) goods, with every item value between \(\frac{1}{64}\) and \(1\).

% In this subsection, we complement the result by establishing the upper bound of the competitive ratio when the algorithm has the accurate MIV prediction.
% In particular, we will construct an instance with two agents and at most $518$ goods, and the value of a single good is at least $r=\frac{1}{64}$ and at most 1. 

\begin{theorem}\label{thm::MIV-impossible}
    For any $\varepsilon>0$ and any $\rho \geq (1-\frac{1}{64\times 518}) + \varepsilon$, there is no deterministic algorithm that achieves competitive ratio at least $\rho$, even for two agents.
\end{theorem}

Before presenting the formal proof, we first explain the high-level idea of the construction. Define $s^t_i = 2(B^t_i+U^t_i) - V^t_i$, the \emph{safety distance} of agent $i$ at the end of round $t$.
Since $n=2$, it is easy to verify that in round $t$, agent $i$ satisfies PROP1 if and only if $s^t_i \geq 0$. The goal of the adversary is therefore to force the $s^t_i$ of some agent below zero. 

\begin{figure}[htbp]
  \centering
  \includegraphics[width=0.85\linewidth]{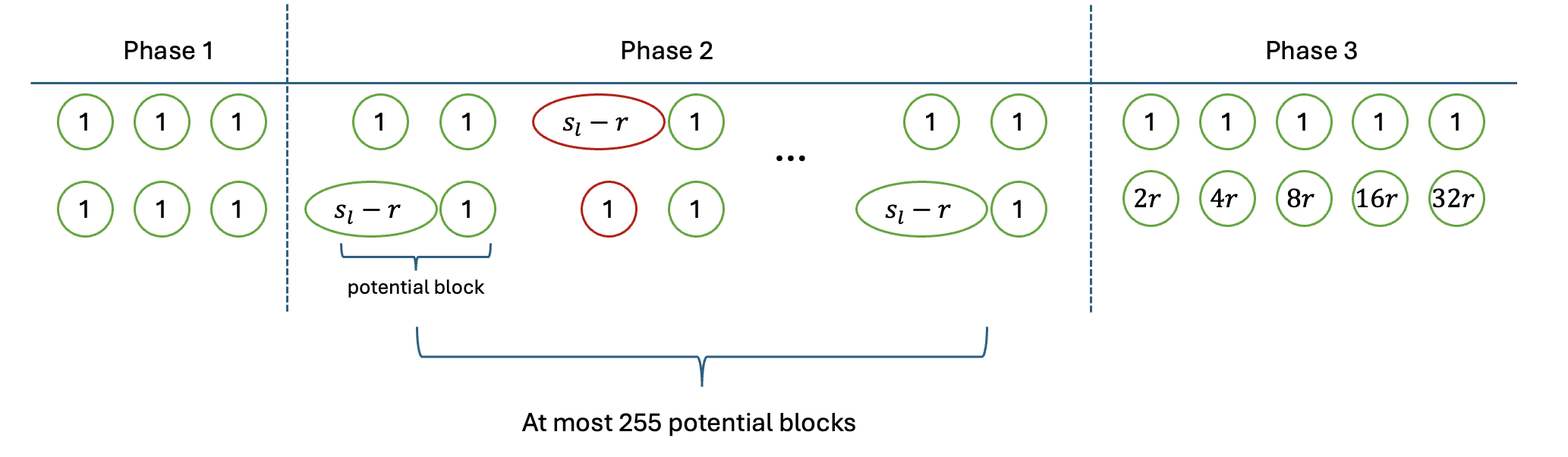}
  \caption{Illustration of the counterexample construction in the proof of Theorem~\ref{thm::MIV-impossible}.}
  \label{fig:not-guarantee-1-PROP1}
\end{figure}

The constructed instance can be divided into three phases. In Phase 1, there are three goods with value 1 each for both agents. Hence, an algorithm that always returns PROP1 allocation must not allocate all three goods to a single agent, which ensures $U^t_1=U^t_2=1$ for every later round $t$, given that every later good has value at most 1. 
Then in Phase 2, there will be at most 255 \emph{potential blocks}, each of which has exactly two goods. We design values of each block adaptively against the algorithm, so that after the allocation of each potential block, the potential $s^t_1+s^t_2$ is decreased by $r$.
Moreover, potential blocks also ensures that right after Phase 2, one agent has safety distance at most $r$ (suppose agent $L$), and the other agent has safety distance at most 4 (suppose agent $H$).
Finally in Phase 3, the adversary construct a \emph{trap block} with five goods, and their values are designed in the way that each of them must be allocated to agent $L$; otherwise, agent $L$ fails PROP1.
However, each of these good has value 1 for agent $H$, so that the algorithm can never satisfy both agents.

\begin{proof}[Proof of Theorem~\ref{thm::MIV-impossible}]
For a contradiction, assume that there exists a deterministic algorithm $ALG$ that achieves competitive ratio at least $\rho$. In the proof, we call the agent with a higher $s$ value as agent $H$, and the other agent as agent $L$.
Agents may switch the roles $H$ and $L$ after the allocations of goods in Phase 2.

In Phase 1, there are three goods $g_1,g_2,g_3$ with value 1 each for both agents. After $ALG$ has decided the allocation of $g_1,g_2,g_3$, agent $H$ receives two goods and agent $L$ receives one good. Now $s^3_H=3$ and $s^3_L=1$.
Then the adversary uses at most 255 potential blocks to ensure that at some moment (suppose round $t_2$), the safety distances of agents become $s^{t_2}_H \leq 4$ and $s^{t_2}_L \leq r$. We now treat this as an assumption, which we will prove later on. Right after this moment, the construction enters Phase 3.

In Phase 3, we fix the name of the agents and no longer update $H$ and $L$. There are five goods $g_{t_2+1}, g_{t_2+2}, g_{t_2+3}, g_{t_2+4}, g_{t_2+5}$. For agent $H$, her value is $v_H(g_{t_2+j})=1$ for every $j=1,\ldots, 5$. And for agent $L$, her value is $v_L(g_{t_2+j}) = 2^j\cdot r$.
We claim that $ALG$ must assign the first four goods to agent $L$. Note that $v_L(g_{t_2+1}) = 1/32$ and $s^{t_2}_L \leq r=1/64$. If $g_{t_2+1}$ is allocated to $H$, then $s^{t_2+1}_L \leq -r$, contradicting the assumption of $ALG$.
After assigning $g_{t_2+1}$ to $L$, we have $s^{t_2+1}_L \leq 3/64$. Since $v_L(g_{t_2+2}) = 1/16$, again $ALG$ must allocate $g_{t_2+2}$ to agent $L$; otherwise, $s^{t_2+2}_L \leq -r$, a contradiction.
Similarly, $ALG$ must assign both $g_{t_2+3}$ and $g_{t_2+4}$ to agent $L$. Right after round $t_2+4$, we have $s^{t_2+4}_H \leq 0$ since $s^{t_2}_H \leq 4$ and agent $H$ does not receive four goods with value 1 each for her. 
For agent $L$, we have $s^{t_2+4}_L \leq 31/64$.
Now since $v_L(g_{t_2+5}) = 1/2>s^{t_2+4}_L$ and $s^{t_2+4}_H \leq 0$, the agent who does not receive $g_{t_2+5}$ would have safety distance at most $-r$. Since there are at most 518 goods and each good has value at most 1, the approximation guarantee of the agent with safety distance at most $-r$ is at most $1-\frac{1}{64\times 518}$.

We now show that the adversary can construct Phase 2 that meets the properties described above. Each potential block has two goods $g^1,g^2$, and $g^1$ comes first. The value of agents are $v_H(g^1)=1$, $v_L(g^1)=s_L-r$, and $v_H(g^2)=v_L(g^2)=1$. Notice that the value of $g^1$ is adaptive and $s_L$ here is the safety distance of agent $L$ right before allocating $g_1$. 
% With slight abuse of notations, when presenting the safety distances, we omit the superscript for the round and use $s_H$ and $s_L$. Nonetheless, we will explicitly point out the moment to aviod confusion.
We claim that once the $s_H \leq 3$ and $s_L\leq 1$ at the beginning of each potential block, these upper bounds still hold after $ALG$ allocating the two goods (without entering Phase 3). 

We prove the claim by induction. Since $s^3_H=3$ and $s^3_L=1$, the base case holds trivially. Assume that the construction reaches the potential block with $g^1,g^2$. Define $h=s_H$ and $\ell =s_L$, where $s_H$ and $s_L$ here are safety distances of agents right before allocating $g^1$.
By the induction hypothesis, $h \leq 3$ and $\ell \leq 1$. By the definitions of values, we have $v_H(g^1)=1, v_L(g^1)=\ell-r$. 
If $ALG$ allocates $g^1$ to $H$, then the safety distances become $s_H\leq 4$ and $s_L=r$, so that the adversary can use Phase 3 to force some agent having safety distance at most $-r$.
Thus, $ALG$ must allocate $g^1$ to $L$. After allocating $g_1$, we have $s_H = h- 1$ and $s_L=2\ell - r$.
For $g^2$, if $ALG$ allocates it to agent $H$, the safety distances become $s'_H = h$ and $s'_L = 2\ell -1 -r$.
And if $ALG$ allocates it to agent $L$, the safety distance become $s''_H=h-2$ and $s''_L=2\ell + 1-r$.
Because $h\leq 3$ and $\ell \leq 1$, we have
\[
\max\{s'_H,s'_L\} \leq 3, \min\{s'_H,s'_L \} \leq 1 \text{ and } \max\{s''_H,s''_L\} \leq 3, \min\{s''_H,s''_L \} \leq 1.
\]
After the allocation of $g^2$, we rename the agents based on their safety distances at that moment. The above inequalities ensure that the safety distance of $H$ and of $L$ are at most 3 and at most 1, respectively. By induction, we prove the claim. 

Moreover, one can verify that if $ALG$ allocates a potential block without entering Phase 3, the total safety distance of two agents is decreased by $r$. Indeed, to avoid entering Phase 3, $ALG$ must allocate $g^1$ to $L$, which decreases the total safety distance by at least $r$.
Since both agent value $g^2$ at 1, the allocation of $g^2$ keeps the total safety distance unchanged.

Recall that initially, the total safety distance is $s^3_H+ s^3_L=4$. Since $r=1/64$ and each potential block decreases the total safety distance by $r$, $ALG$ can safely handle at most 255 potential blocks without entering Phase 3. 
After the allocating all 255 potential blocks, the total safety distance is at most $r$, and hence, the adversary can use Phase 3 to force the safety distance of some agent becoming at most $-r$.

Overall, the adversary uses at most $3+2\times 255+5=518$ goods and force the safety distance of some agent becoming negative. Notice that the values of goods for both agents are integer multiple of $r$.
Once the safety distance becomes negative, we have that it is at most $-r$, which then gives the desired upper bound of the competitive ratio of $ALG$. 
\end{proof}

\section{A Consistent and Robust Algorithm}
\label{subsec:simultaneous-robustness}

Section~\ref{sec::accurate-MIV} assumes that the maximum item value predictions are accurate. Nonetheless, in this section, we study the setting when the predictions can fail. Our goal here is to obtain both earlier algorithmic guarantees through one allocation rule: $\frac{1}{2}$-PROP1 when predictions are accurate and $\Omega(\frac{1}{\log nm})$-PROP1 when the predictions are arbitrarily bad. 

Since allocations are irrevocable, the algorithm cannot construct the allocations of Algorithms~\ref{alg:no-informaton} and \ref{alg:MIV-0.5} separately and choose between them at termination. 
A method that switches only after detecting an inaccurate prediction also does not address every error: a prediction larger than the realized MIV may never be contradicted by an arriving good. 
We therefore design an algorithm that maintains prediction-based quantities and quantities that depend only on observed values from the beginning of the sequence. It modifies the reciprocal potential in Algorithm~\ref{alg:MIV-0.5} and combines it with the exponential potential in Algorithm~\ref{alg:no-informaton}. 
Each arriving good is assigned according to its effect on this combined potential.
The following theorem states the both consistent and robust guarantees of our proposed algorithm.

\begin{theorem}
\label{thm:simultaneous-consistency-robustness}
For any $n\geq 2$, there exists a deterministic online algorithm satisfying the following properties.
\begin{enumerate}
    \item[(i)] when the predictions are accurate, it achieves competitive ratio of $\frac{1}{2}$.
    % If $p_i=\max_{\ell\in[m]}v_i(g_\ell)$ for every agent $i$, then the terminal allocation is $1/2$-PROP1. More precisely, every agent with $V_i^m>0$ satisfies
    % \[
    % B_i^m+U_i^m>\frac{V_i^m}{2n}.
    % \]
    \item[(ii)] when the predictions are arbitrarily bad,
    it achieves competitive ratio of at least $\frac{3}{184\log(nm)}$.
    % every arbitrary nonnegative vector $(p_1,\ldots,p_n)$ and every round $t\geq 1$,
    % \[
    % \rho^t\geq
    % \frac{1}{4+\frac{16}{3}\log\!\bigl(n(100n+1)t(t+1)\bigr)}
    % \geq \frac{3}{184\log(nt)}.
    % \]
\end{enumerate}
% The algorithm stores a constant number of real quantities for every agent and uses $O(n)$ arithmetic operations and exponential evaluations in every round.
\end{theorem}

We now describe the algorithm. For ease of notation, let $\beta=\frac{1}{2n},\theta=\frac{1}{4n},\gamma=\frac{3}{4},\eta=\frac{1}{100n},$ and $Q=n+\frac{1}{100}$.
We continue to use the function $g(x)=\frac{1}{x(x+1)}$. 
The algorithm receives the predictions $p_1,\ldots,p_n$ before the first good arrives.

For each agent \(i\) with \(p_i>0\), the algorithm initially uses \(p_i\) to normalize her values. 
We say that \(p_i\) is \emph{active} as long as every value revealed for agent \(i\) is at most \(p_i\). 
If \(v_i(g_t)>p_i\) in some round \(t\), then the algorithm deactivates \(p_i\) before computing the quantities in that round. 
Once \(p_i\) has been deactivated, it remains inactive and is no longer used by the algorithm. Consequently, whenever \(p_i\) remains active after the value vector of \(g_t\) is revealed, the normalized value $y_i^t=\frac{v_i(g_t)}{p_i}$ is well defined and belongs to $[0,1]$.
The algorithm maintains two groups of quantities. The first group depends on \(p_i\) and is used to recover the 0.5 competitive ratio when all predictions are exact. 
The second group depends only on values that have already been revealed and provides the logarithmic guarantee when the predictions are arbitrary.

We first describe the quantities that use \(p_i\). While \(p_i\) is active, the algorithm maintains $ \delta_i^t=\frac{\beta V_i^t-B_i^t}{p_i}$.
Thus, \(p_i\delta_i^t\) is the difference between the target value \(\beta V_i^t\) and the value \(B_i^t\) received by agent \(i\) by the end of round \(t\).
The algorithm also maintains the same three statuses as Algorithm~\ref{alg:MIV-0.5}. 
Initially, $ \delta_i^0=0, \textsc{Sta}_i^0=\textsc{Sta1}$. The status is updated after \(g_t\) is allocated. If \(y_i^t<1\), then the status remains unchanged. 
If \(y_i^t=1\) and the status before round \(t\) is either \(\textsc{Sta1}\) or \(\textsc{Sta2}\), then the new status is \(\textsc{Sta2}\) when agent \(i\) receives \(g_t\), and it is \(\textsc{Sta3}\) when she does not receive \(g_t\). 
Once the status becomes \(\textsc{Sta3}\), it remains \(\textsc{Sta3}\) in every later round. These are exactly the update rules used in Algorithm~\ref{alg:MIV-0.5}.
Recall
$h(\textsc{Sta1})=1-\beta,
h(\textsc{Sta2})=0,$ and $
h(\textsc{Sta3})=1
$.

For every agent whose prediction is active at the end of round \(t\), let $ q_i^t=h(\textsc{Sta}_i^t)-\delta_i^t,$ and $ \bar q_i^t=q_i^t+\beta$.
The proof will maintain \(q_i^t>0\). The quantity \(\bar q_i^t\) is used when the algorithm decides the recipient of the current good.
After the value vector of \(g_t\) has been revealed and predictions have been deactivated, an agent \(i\) whose prediction remains active is called \emph{special} in round \(t\) if $y_i^t=1$ and $\textsc{Sta}_i^{t-1}\in\{\textsc{Sta1},\textsc{Sta2}\}$.
An agent whose prediction remains active but who does not satisfy these two conditions is called \emph{ordinary}. An agent whose prediction is inactive is neither special nor ordinary.

Consider an ordinary agent \(i\). If she does not receive \(g_t\), then the value of \(\bar q_i\) after the round would be $ \bar q_i^{t,-}=\bar q_i^{t-1}-\beta y_i^t$.
If she receives \(g_t\), then the corresponding value would be
$ \bar q_i^{t,+}=\bar q_i^{t-1}+(1-\beta)y_i^t. $
Define $ G_i^t := \frac{1}{\bar q_i^{t,-}} - \frac{1}{\bar q_i^{t,+}}$, which is the exact difference between the reciprocal term \(1/\bar q_i\) when agent \(i\) does not receive \(g_t\) and the same reciprocal term when she receives \(g_t\). 
A larger \(G_i^t\) means that allocating \(g_t\) to agent \(i\) avoids a larger increase in this reciprocal term.
For a special agent, the change in \(q_i\) is independent of the recipient of \(g_t\): if the pre-round status is \(\mathrm{STA1}\), then \(q_i\) remains unchanged, and if the pre-round status is \(\mathrm{STA2}\), then \(q_i\) increases by \(1-\beta\). 
Hence, the current allocation decision does not create a receive-versus-miss difference for a special agent. We therefore let $ G_i^t=0$  for every special agent. 
We also let \(G_i^t:=0\) whenever \(p_i\) is inactive.

We next describe the quantities that do not use the predictions. For every agent \(i\), the algorithm maintains the same quantities as Algorithm~\ref{alg:no-informaton}. 
In particular, $ H_i^t:=\max_{\ell\le t}v_i(g_\ell) $ is the largest value observed for agent \(i\) by round \(t\), and \(\tau_i^t\) is the number of strict positive increases of this running maximum, where the first positive value is counted as the first increase. 
The algorithm also maintains $D_i^t:=\theta V_i^t-B_i^t$.
In round \(t\), after the complete value vector of \(g_t\) is revealed, the algorithm first updates \(H_i^t\) and \(\tau_i^t\) for every agent. At this point, \(g_t\) has not yet been allocated, so the quantity used in the subsequent calculation remains \(D_i^{t-1}\). 
If \(H_i^t>0\), define $ z_i^t:=\frac{v_i(g_t)}{H_i^t}$.
Because \(H_i^t\) includes the value of the current good, we have \(z_i^t\in[0,1]\). 
The algorithm then defines
$$ \Delta_i^t := g(\tau_i^t) \exp\!\left( \gamma\frac{D_i^{t-1}}{H_i^t} \right) \left[ \exp(\gamma\theta z_i^t) - \exp\!\left(-\gamma(1-\theta)z_i^t\right) \right]. $$
If \(H_i^t=0\), then agent \(i\) assigns value zero to every good that has arrived by round \(t\), and we let $ z_i^t:=0$ and $ \Delta_i^t=0$.
The interpretation of \(\Delta_i^t\) follows from the update of \(D_i\). If agent \(i\) does not receive \(g_t\), then
$ D_i^t=D_i^{t-1}+\theta v_i(g_t). $
If she receives \(g_t\), then $ D_i^t=D_i^{t-1}-(1-\theta)v_i(g_t). $
Consequently, \(\Delta_i^t\) is the difference between
$ g(\tau_i^t) \exp\!\left( \gamma\frac{D_i^t}{H_i^t} \right) $ when agent \(i\) does not receive \(g_t\) and the same expression when she receives \(g_t\). 
Thus, a larger \(\Delta_i^t\) means that assigning \(g_t\) to agent \(i\) avoids a larger increase in this exponential expression.

The two quantities \(G_i^t\) and \(\Delta_i^t\) measure the effect of assigning \(g_t\) to agent \(i\) on two explicitly defined expressions. 
The first uses the active prediction \(p_i\), while the second uses only values observed by round \(t\). Algorithm~\ref{alg:simultaneous-robust} combines these effects by allocating \(g_t\) to an agent
$$ i_t\in \arg\max_{i\in N} \left\{ G_i^t+\eta\Delta_i^t \right\}. $$
Ties are broken according to a fixed deterministic ordering of the agents. After selecting \(i_t\), the algorithm updates \(D_i^t\) for every agent, and it updates \(\delta_i^t\), \(\mathrm{STA}_i^t\), \(q_i^t\), and \(\bar q_i^t\) for every agent whose prediction remains active. The formal procedure is given in Algorithm~\ref{alg:simultaneous-robust}. 

\begin{algorithm}[H]
\caption{Simultaneously consistent and robust allocation with MIV predictions}
\label{alg:simultaneous-robust}
\begin{algorithmic}[1]
\STATE \textbf{Input:} Predictions $p_1,\ldots,p_n\geq 0$ and a fixed deterministic ordering of the agents for tie-breaking.
\STATE Initialize $D_i^0=H_i^0=\tau_i^0=0$ and $A_i=\varnothing$ for every $i\in[n]$.
\STATE For every $i\in[n]$, if $p_i>0$, mark $p_i$ as active and initialize $ \delta_i^0=0, \textsc{Sta}_i^0=\textsc{Sta1}, q_i^0=1-\beta$, and $\bar q_i^0=1$.
If $p_i=0$, mark $p_i$ as inactive.
\WHILE{a new good $g_t$ arrives}
\STATE Observe $v_i(g_t)$ for every $i\in[n]$.
\STATE For every active $p_i$ satisfying $v_i(g_t)>p_i$, mark $p_i$ as inactive. \label{alg::step::robuts-deactivate}
% Once inactive, $p_i$ remains inactive.
\STATE For every $i\in[n]$, update
$
H_i^t\leftarrow\max\{H_i^{t-1},v_i(g_t)\}$ and $
\tau_i^t\leftarrow
\tau_i^{t-1}
+
\mathds{1}
_{\{v_i(g_t)>H_i^{t-1}\}}$.
\STATE For every $i$ whose prediction $p_i$ remains active, compute
$
y_i^t=\frac{v_i(g_t)}{p_i},
$
determine whether $i$ is special or ordinary, and compute $G_i^t$ according to the definitions above. Set $G_i^t=0$ for every other agent.
\STATE For every $i\in[n]$, compute $z_i^t$ and $\Delta_i^t$.
\STATE Allocate $g_t$ to
\[
i_t\in
\arg\max\limits_{i\in[n]}
\left\{G_i^t+\eta\Delta_i^t \right\},
\]
breaking ties according to the fixed ordering, and update $
A_{i_t}\leftarrow A_{i_t}\cup{g_t}$.
\STATE For every $i\in[n]$, update
$ D_i^t
\leftarrow
D_i^{t-1}
+ \theta v_i(g_t) - \mathds{1}_{\{i=i_t\}}v_i(g_t)$.
\STATE For every $i$ whose prediction $p_i$ remains active, update
$
\delta_i^t
\leftarrow
\delta_i^{t-1}
+
\beta y_i^t
-
\mathds{1}_{\{i=i_t\}}y_i^t.
$
Update $\textsc{Sta}_i^t$ according to the rules stated above, and
$
q_i^t
\leftarrow
h(\textsc{Sta}_i^t)-\delta_i^t,
\bar q_i^t
\leftarrow
q_i^t+\beta.
$
\ENDWHILE
\STATE \textbf{Output:} The allocation $A=(A_1,\ldots,A_n)$.
\end{algorithmic}
\end{algorithm}

We next analyze Algorithm~\ref{alg:simultaneous-robust}. 
For any round $t$, let $P^t$ be the set of agents whose predictions remain active after the deactivation step in round $t$ (Line~\ref{alg::step::robuts-deactivate}).
At the end of round $t$, after $g_t$ has been allocated and all quantities have been updated, define the augmented potential
\[
\Gamma^t=
\sum_{i\in P^t}\frac{1}{\bar q_i^t}
+\eta\sum_{i:H_i^t>0}g(\tau_i^t)
\exp\!\left(\gamma\frac{D_i^t}{H_i^t}\right)
+\eta\sum_{i=1}^n\frac{1}{\tau_i^t+1}.
\]
The first sum contains the reciprocal terms based on the predictions that are still active. The second sum contains the exponential terms based only on values revealed by round $t$. 
The final sum compensates for a possible increase in the exponential term when the running maximum $H_i^t$ changes.
To separate the updates that occur before the allocation of $g_t$ from those that occur after the allocation, define
\[
\widehat\Gamma^t
=
\sum_{i\in P^t}\frac{1}{\bar q_i^{t-1}}
+
\eta
\sum_{i:H_i^t>0}
g(\tau_i^t)
\exp\!\left(
\gamma\frac{D_i^{t-1}}{H_i^t}
\right)
+
\eta
\sum_{i=1}^n
\frac{1}{\tau_i^t+1}.
\]
Thus, $\widehat\Gamma^t$ is evaluated after every prediction satisfying $v_i(g_t)>p_i$ has been deactivated and after $H_i^t$ and $\tau_i^t$ have been updated, but before $g_t$ has been allocated. Accordingly, its reciprocal terms use $\bar q_i^{t-1}$, while its exponential terms use the pre-allocation quantity $D_i^{t-1}$ together with the updated values $H_i^t$ and $\tau_i^t$.

\begin{lemma}
\label{lem:robust-preparation-domain}
Fix a round $t\geq 1$. Suppose that for any $i\in P^{t-1}$, $q_i^{t-1}>0$.
Then $
\widehat\Gamma^t\leq\Gamma^{t-1}$ holds.
If $ \Gamma^{t-1}\leq Q$,
then 
\begin{itemize}
    \item every $i\in P^t$ satisfies
$\bar q_i^{t-1}\geq\frac{1}{Q}$ and $q_i^{t-1}\geq\frac{1}{Q}-\beta>0$;
\item for every ordinary agent $i\in P^t$, $\bar q_i^{t,-} \geq \frac{1}{Q}-\beta>0$, and $ \bar q_i^{t,+} \geq \frac{1}{Q} >0$;
\item for every special agent $i\in P^t$, either possible allocation of $g_t$ gives $q_i^t\geq q_i^{t-1}>0$.
\end{itemize}

% For any round $t\geq 1$, it holds that $ \widehat\Gamma^t\leq\Gamma^{t-1}$.
% Moreover, if $\widehat\Gamma^t\leq Q$, then every agent $i\in P^t$ satisfies
% \[
% \bar q_i^{t-1}\geq\frac{1}{Q},\qquad
% q_i^{t-1}\geq\frac{1}{Q}-\beta>0,
% \qquad
% \bar q_i^{t,-}\geq\frac{1}{Q}-\beta>0.
% \]
% Consequently, every score in Algorithm~3 is well-defined before it is used.
\end{lemma}
\begin{proof}
We first prove $\widehat\Gamma^t\leq\Gamma^{t-1}$. Consider an agent $i\in P^{t-1}\setminus P^t$. The algorithm deactivates $p_i$ in round $t$, and the term $ \frac{1}{\bar q_i^{t-1}}$ is removed. 
By the assumption $q_i^{t-1}>0$, we have $ \bar q_i^{t-1}=q_i^{t-1}+\beta>0$. 
Hence, the removed term is positive, so its removal cannot increase the potential.

We next compare, for each agent $i$, the contribution to $\Gamma^{t-1}$ with the corresponding contribution to $\widehat\Gamma^t$. 
During this comparison, the current good has not yet been allocated. Therefore, the value in the numerator of the exponential term remains $D_i^{t-1}$ on both sides. 
Suppose first that $v_i(g_t)\leq H_i^{t-1}$.
Then the current good is not a new strict maximum for agent $i$, so $ H_i^t=H_i^{t-1}$ and $\tau_i^t=\tau_i^{t-1}$.
If $H_i^{t-1}>0$, it follows that
\[
g(\tau_i^t)
\exp\!\left(
\gamma\frac{D_i^{t-1}}{H_i^t}
\right)
=
g(\tau_i^{t-1})
\exp\!\left(
\gamma\frac{D_i^{t-1}}{H_i^{t-1}}
\right) \text{ and } \frac{1}{\tau_i^t+1}
=
\frac{1}{\tau_i^{t-1}+1}.
\]
Thus, these two terms associated with agent $i$ are unchanged during the transition from $\Gamma^{t-1}$ to $\widehat\Gamma^t$.
If $H_i^{t-1}=0$, then $v_i(g_t)=0$, so $H_i^t=0$ and $\tau_i^t=0$. In this case, the exponential term is absent both before and after the pre-allocation update, and the term $1/(\tau_i+1)$ is also unchanged.

Suppose now that $ v_i(g_t)>H_i^{t-1}$.
Then $v_i(g_t)$ is a new strict positive maximum. We distinguish whether this is the first positive value observed for agent $i$.
Assume first that $\tau_i^{t-1}\geq 1$. Before the update, the sum of the two relevant terms are
\[
\eta g(\tau_i^{t-1})
\exp\!\left(
\gamma\frac{D^{t-1}_i}{H^{t-1}_i}
\right)
+
\frac{\eta}{\tau^t_i}.
\]
After the update, but before allocating $g_t$, they are
\[
\eta g(\tau^t_i)
\exp\!\left(
\gamma\frac{D^{t-1}_i}{H^t_i}
\right)
+
\frac{\eta}{\tau^t_i+1}.
\]
Suppose first that $D^{t-1}_i\geq 0$. Since $H^t_i>H^{t-1}_i$, we have $\frac{D^{t-1}_i}{H^t_i} \leq \frac{D^{t-1}_i}{H^{t-1}_i}$, and therefore,
$
\exp(
\gamma\frac{D^{t-1}_i}{H^t_i})
\leq
\exp(
\gamma\frac{D^{t-1}_i}{H^{t-1}_i}).
$
Moreover, $g(\tau^t_i)<g(\tau^{t-1}_i)$ and  $\frac{1}{\tau^t_i+1}<\frac{1}{\tau^t_i}$.
It follows that the sum of the exponential term and the final compensation term does not increase.

Suppose next that $D<0$. Then $\exp(\gamma\frac{D^{t-1}_i}{H^t_i})<1$, so the new exponential term is strictly smaller than $\eta g(\tau^t_i)$.
Hence, the new sum is strictly smaller than $ \eta g(\tau^t_i)+\frac{\eta}{\tau^t_i+1}$.
Using $g(\tau^t_i)=\frac{1}{\tau^t_i(\tau^t_i+1)}$,
we have
$
\eta g(\tau^t_i)+\frac{\eta}{\tau^t_i+1}
=
\frac{\eta}{\tau^t_i}
$.
The old compensation term is already equal to $\frac{\eta}{\tau^t_i}$, and the old exponential term is positive. Therefore, the new sum is strictly smaller than the old sum.

It remains to consider the case that $g_t$ gives the first positive value. In this case, $H_i^{t-1}=0$ and $\tau_i^{t-1}=0$.
All values revealed for agent $i$ before round $t$ were zero. Consequently,
$
V_i^{t-1}=B_i^{t-1}=D_i^{t-1}=0$.
Before the update, the exponential term is absent and the compensation term is $\eta$. After the update, $\tau_i^t=1$, and the new exponential term is
$ \eta g(1)\exp(0)=\frac{\eta}{2}$.
The new compensation term is also $\frac{\eta}{2}$. Thus, the total contribution remains equal to $\eta$.

These comparisons hold separately for every agent. Summing them and including the removal of deactivated predictions gives $\widehat\Gamma^t\leq\Gamma^{t-1}$.

Next, we prove the positivity statements. Suppose that $\Gamma^{t-1}\leq Q$.
The first part of the proof gives $ \widehat\Gamma^t \leq \Gamma^{t-1} \leq Q$. Fix $i\in P^t$. Since $P^t\subseteq P^{t-1}$, the term $\frac{1}{\bar q_i^{t-1}}$
appears in $\widehat\Gamma^t$. All terms in $\widehat\Gamma^t$ are nonnegative. Therefore, $\frac{1}{\bar q_i^{t-1}}
\leq
\widehat\Gamma^t
\leq Q$, which implies $ \bar q_i^{t-1}\geq\frac{1}{Q}$.
Since $q_i^{t-1}=\bar q_i^{t-1}-\beta$, we obtain $q_i^{t-1} \geq \frac{1}{Q}-\beta$.
Moreover, $Q=n+\frac{1}{100}<2n=\frac{1}{\beta}$.
Hence, $\frac{1}{Q}-\beta>0$. Because $i\in P^t$, the prediction remains active after observing $g_t$. Therefore, $0\leq y_i^t\leq 1$.
If $i$ is ordinary, then $\bar q_i^{t,-}=
\bar q_i^{t-1}-\beta y_i^t >0$ and $\bar q_i^{t,+}=
\bar q_i^{t-1}+(1-\beta)y_i^t > 0 $. 
If $i$ is special, we have  $q_i^t\in
\left\{
q_i^{t-1},
q_i^{t-1}+1-\beta
\right\}$, and thus, $q_i^t\geq q_i^{t-1}>0$.
\end{proof}

The next lemma proves that the potential never exceeds its initial upper bound. For the analysis only, consider the hypothetical update in which every agent does not receive $g_t$. This is not a feasible allocation. It is only a convenient way to write the changes of the different terms. After calculating the potential in this hypothetical update, assigning $g_t$ to one agent subtracts the exact receive-versus-miss reduction associated with that agent.

\begin{lemma}
\label{lem:robust-augmented-invariant}
For any round $t\geq 0$, $\Gamma^t\leq Q$ holds.
Moreover, every $i\in P^t$ satisfies $q_i^t
\geq
\frac{1}{Q}-\beta
>0$.
\end{lemma}

\begin{proof}
We proceed by induction on $t$. At $t=0$, every $i\in P^0$ satisfies $q_i^0=1-\beta$ and $\bar q_i^0=1$.
Thus, every active prediction contributes one to the first sum, and there are at most $n$ such terms. Since $H_i^0=0$ for every agent $i$, there is no exponential term. Finally, we have $\eta\sum_{i=1}^n\frac{1}{\tau_i^0+1}
=
n\eta
=
\frac{1}{100}$, and thus, $\Gamma^0
\leq
n+\frac{1}{100}
=
Q$ and $q_i^0=1-\beta>0$.

Fix $t\geq 1$, and suppose that the conclusion holds at the end of round $t-1$. Lemma~\ref{lem:robust-preparation-domain} gives $\widehat\Gamma^t
\leq
\Gamma^{t-1}
\leq Q$.
The same lemma proves that after allocating $g_t$, all denominators are positive. 
Hence, $\Gamma^t$ is well-defined for every possible recipient of $g_t$.

We first calculate the increase in the reciprocal term of an ordinary agent under the hypothetical update in which she does not receive $g_t$. Fix an ordinary $i\in P^t$, and define
\[
I_i^t =
\frac{1}{\bar q^{t-1}_i-\beta y^{t}_i}
-
\frac{1}{\bar q^{t-1}_i}.
\]
By definition, we have
\[
\begin{aligned}
G_i^t =
\frac{1}{\bar q^{t-1}_i-\beta y^t_i} -\frac{1}{\bar q^{t-1}_i+(1-\beta)y^t_i} = \frac{y^t_i} {(\bar q^{t-1}_i-\beta y^t_i)(\bar q^{t-1}_i+(1-\beta)y^t_i)}.\end{aligned}
\]
If $y^t_i>0$, division gives
\[
\begin{aligned}
\frac{I_i^t}{G_i^t}
=
\beta
\frac{\bar q^{t-1}_i+(1-\beta)y^t_i}{\bar q^{t-1}_i}=
\beta
+
\beta(1-\beta)\frac{y^t_i}{\bar q^{t-1}_i},
\end{aligned}
\]
which implies
\[
I_i^t
=
\left(
\beta
+
\beta(1-\beta)
\frac{y_i^t}{\bar q_i^{t-1}}
\right)
G_i^t.
\]
If $y^t_i=0$, then $I_i^t=G_i^t=0$,
so the same identity remains valid.

We next verify the reciprocal change of every special agent. 
The following table lists the change in $h(\textsc{Sta}_i)$, the change in $\delta_i$, and the resulting change in $q_i=h(\textsc{Sta}_i)-\delta_i$:
\[
\begin{array}{c|c|c|c|c}
\text{pre-round status}
&
\text{agent receives }g_t
&
h(\textsc{Sta}_i^t)-h(\textsc{Sta}_i^{t-1})
&
\delta_i^t-\delta_i^{t-1}
&
q_i^t-q_i^{t-1}
\\ \hline
\textsc{Sta1}
&
\text{yes}
&
-(1-\beta)
&
-(1-\beta)
&
0
\\
\textsc{Sta1}
&
\text{no}
&
\beta
&
\beta
&
0
\\
\textsc{Sta2}
&
\text{yes}
&
0
&
-(1-\beta)
&
1-\beta
\\
\textsc{Sta2}
&
\text{no}
&
1
&
\beta
&
1-\beta
\end{array}
\]
Thus, $q^t_i-q^{t-1}_i$ either remains unchanged or increases by $1-\beta$. Consequently, the reciprocal term of a special agent either remains unchanged or decreases. 
% This also shows why setting
% \[
% G_i^t=0
% \]
% for a special agent is exact.

We now calculate the increase in the exponential term when an agent does not receive $g_t$. For every agent $i$ with $H_i^t=0$, let $J_i^t=0$.
For every agent with $H_i^t>0$, let
\[
J_i^t =
g(\tau_i^t)
\exp\!\left(
\gamma\frac{D_i^{t-1}}{H_i^t}
\right)
\left[
\exp(\gamma\theta z_i^t)-1
\right].
\]

If $z_i^t=0$, then $J_i^t=\Delta_i^t=0$.
Suppose that $z_i^t>0$. Then we have
\[
\frac{J_i^t}{\Delta_i^t}
=
\frac{
\exp(\gamma\theta z_i^t)-1
}{
\exp(\gamma\theta z_i^t)
-
\exp(-\gamma(1-\theta)z_i^t)
}.
\]
Multiplying the numerator and denominator by $\exp(-\gamma\theta z_i^t)$
gives
\[
\frac{J_i^t}{\Delta_i^t}
=
\frac{
1-\exp(-\gamma\theta z_i^t)
}{
1-\exp(-\gamma z_i^t)
}.
\]

For every $x\geq 0$, $1-e^{-x}\leq x$.
Therefore, $1-\exp(-\gamma\theta z_i^t)
\leq
\gamma\theta z_i^t$.
The function $f(x)=1-e^{-x}$ is concave and satisfies $f(0)=0$. Since $0\leq\gamma z_i^t\leq\gamma$, concavity gives $1-\exp(-\gamma z_i^t) \geq z_i^t(1-\exp(-\gamma))$.
Consequently, we have $\frac{J_i^t}{\Delta_i^t}
\leq
\frac{\gamma\theta}{1-\exp(-\gamma)}$.
Moreover, it holds that $\exp(\frac{3}{4})>2$,
so $1-\exp(-\gamma)>\frac{1}{2}$.
Using $\gamma=\frac{3}{4}$ and $\theta=\frac{1}{4n}$,
we have $J_i^t \leq \frac{3}{8n}\Delta_i^t$.

For ease of notation, let $ \mathcal{R}^t = G_{i_t}^t+\eta \Delta^{t}_{i_t} = \max_{j\in[n]} \left\{ G_j^t+\eta \Delta^{t}_{j} \right\}$.
Thus, $\mathcal{R}^t$ is the combined reduction associated with the recipient selected by Algorithm~\ref{alg:simultaneous-robust}.
Let $\textsc{Ord}^t$ be the set of ordinary agents whose predictions remain active in round $t$.
For every $i\in\textsc{Ord}^t$, $G_i^t+\eta \Delta^{t}_i\leq\mathcal{R}^t$ and hence $\eta \Delta^{t}_i \leq \mathcal{R}^t-G_i^t$.
For every $i\notin\textsc{Ord}^t$, the algorithm sets $G_i^t=0$, and therefore $\eta \Delta^{t}_i\leq\mathcal{R}^t$.
Summing these inequalities gives
\[
\sum_{i=1}^n \eta \Delta^{t}_i
\leq
n\mathcal{R}^t
-
\sum_{i\in\textsc{Ord}^t}G_i^t.
\]
Define $S^t
=
\sum_{i\in\textsc{Ord}^t}I_i^t
+
\eta\sum_{i=1}^nJ_i^t$.
The reciprocal changes of special agents are nonpositive and are not included in $S^t$. Therefore, the value of the potential under the hypothetical update in which every agent does not receive $g_t$ is at most $\widehat\Gamma^t+S^t$.
Using the identities proved above, we have
\[
\begin{aligned}
S^t
&\leq
\sum_{i\in\textsc{Ord}^t}
\left(
\beta
+
\beta(1-\beta)
\frac{y_i^t}{\bar q_i^{t-1}}
\right)
G_i^t
+
\frac{3}{8n}
\sum_{i=1}^nW_i^t\\
&\leq
\left(
\beta-\frac{3}{8n}
\right)
\sum_{i\in\textsc{Ord}^t}G_i^t +
\beta(1-\beta)
\sum_{i\in\textsc{Ord}^t}
\frac{y_i^t}{\bar q_i^{t-1}}G_i^t
+
\frac{3}{8}\mathcal{R}^t.
\end{aligned}
\]
Since $\beta-\frac{3}{8n}
=
\frac{1}{8n}>0$
and $G_i^t\leq\mathcal{R}^t$, we have $\sum_{i\in\textsc{Ord}^t}G_i^t
\leq
n\mathcal{R}^t$.
Also, it holds that $0\leq y_i^t\leq 1$ and  $G_i^t\leq\mathcal{R}^t$.
Thus, we have
\[
\begin{aligned}
\sum_{i\in\textsc{Ord}^t}
\frac{y_i^t}{\bar q_i^{t-1}}G_i^t
\leq
\mathcal{R}^t
\sum_{i\in\textsc{Ord}^t}
\frac{1}{\bar q_i^{t-1}}\leq
Q\mathcal{R}^t,
\end{aligned}
\]
where the last inequality transition follows because these reciprocal terms are contained in $\widehat\Gamma^t$ that satisfies $\widehat\Gamma^t\leq Q$. Substitution gives
\[
\begin{aligned}
S^t
\leq
\left(
\beta-\frac{3}{8n}
\right)n\mathcal{R}^t
+
\beta(1-\beta)Q\mathcal{R}^t
+
\frac{3}{8}\mathcal{R}^t=
\left(
\beta n+\beta(1-\beta)Q
\right)\mathcal{R}^t.
\end{aligned}
\]
Let $\chi=
\beta n+\beta(1-\beta)Q$.
Substituting $\beta=\frac{1}{2n}$ and $Q=n+\frac{1}{100}$,
we obtain
\[
\begin{aligned}
\chi
=
\frac{1}{2}
+
\frac{1}{2n}
\left(
1-\frac{1}{2n}
\right)
\left(
n+\frac{1}{100}
\right)=
1-\frac{98n+1}{400n^2} <1.
\end{aligned}
\]

We now return from the hypothetical update to the actual allocation. For the exponential term, assigning $g_t$ to agent $i_t$, rather than letting $i_t$ miss it, decreases the potential by exactly $W_{i_t}^t$.
If $i_t$ is ordinary, the corresponding reciprocal term decreases by exactly $G_{i_t}^t$.
If $i_t$ is special or its prediction is inactive, then $G_{i_t}^t=0$, and there is no additional receive-versus-miss change in the reciprocal part. The reciprocal changes of all special agents are already nonpositive and independent of the recipient. Thus, we have
\[
\begin{aligned}
\Gamma^t
\leq
\widehat\Gamma^t
+
S^t
-
\left(
G_{i_t}^t+\eta\Delta^{t}_{i_t} 
\right)
\leq
\widehat\Gamma^t
+
(\chi-1)\mathcal{R}^t\leq
\widehat\Gamma^t,
\end{aligned}
\]
where the last inequality transition is due to $\chi<1$.
Combining the above inequality with Lemma~\ref{lem:robust-preparation-domain} gives
$\Gamma^t \leq \widehat\Gamma^t \leq \Gamma^{t-1} \leq Q$.
This completes the induction.

Finally, if $i\in P^t$, then $\frac{1}{\bar q_i^t} \leq \Gamma^t \leq Q$ holds.
Therefore, $\bar q_i^t\geq\frac{1}{Q}$, and hence $q_i^t = \bar q_i^t-\beta \geq \frac{1}{Q}-\beta >0$.
\end{proof}

We now derive the consistency and robustness guarantees of the algorithm.

\begin{proof}[Proof of Theorem~\ref{thm:simultaneous-consistency-robustness}]
We first prove the guarantee under exact predictions, where $p_i=\max_{\ell\in[m]}v_i(g_\ell)$ for all $i$. For the agent $i$ with $p_i=0$, she trivially satisfies PROP1, and hence, it suffices to focus on the agent with $p_i>0$.

Fix an agent $i$ with $p_i>0$ and hence $V_i^m>0$.
Because the prediction is accurate, the algorithm never deactivates $p_i$. Moreover, at least one good of value $p_i$ appears by the terminal round. 
Hence, the terminal status is either $\textsc{Sta2}$ or $\textsc{Sta3}$. 
By Lemma~\ref{lem:robust-augmented-invariant}, we have $q_i^m>0$. Suppose first that $\textsc{Sta}_i^m=\textsc{Sta2}$.
Since $h(\textsc{Sta2})=0$, we have $q_i^m=-\delta_i^m$, and thus, $\delta_i^m<0$.
By argument similar to the proof of Proposition~\ref{prop::MIV-detla}, we have $p_i\delta_i^m = \beta V_i^m-B_i^m$, and hence, $B_i^m>\beta V_i^m$.
Therefore, $B_i^m+U_i^m>\beta V_i^m$, yielding 0.5 competitive ratio in this case.
Next suppose that $\textsc{Sta}_i^m=\textsc{Sta3}$.
Then at least one good of value $p_i$ was allocated to another agent. Since $p_i$ is the maximum item value, $U_i^m=p_i$ holds.
Moreover, $q_i^m=1-\delta_i^m>0$, so $\delta_i^m<1$.
Consequently, we have $\beta V_i^m-B_i^m = p_i\delta_i^m< p_i= U_i^m$.
Therefore, in this case, we also have $B_i^m+U_i^m>\beta V_i^m$, which results in the desired competitive ratio.
At this stage, we prove the 0.5 competitive ratio when predictions are accurate.

We next prove the guarantee for arbitrary predictions. Fix a round $t$ and an agent $i$.
If $H_i^t=0$, then every good that has arrived by round $t$ has value zero for agent $i$. Thus, agent $i$ satisfies PROP1 trivially.
It suffices to prove the case where $H_i^t>0$.
The corresponding exponential term is contained in $\Gamma^t$. By Lemma~\ref{lem:robust-augmented-invariant}, we have $\eta
g(\tau_i^t)
\exp\!\left(
\gamma\frac{D_i^t}{H_i^t}
\right)
\leq Q$.
Since $H_i^t>0$, at least one strict positive maximum has occurred, and therefore $1\leq\tau_i^t\leq t$.
The function $g(x)=\frac{1}{x(x+1)}$ is decreasing. Hence, $g(\tau_i^t) \geq \frac{1}{t(t+1)}$.
It follows that $\exp\!\left(
\gamma\frac{D_i^t}{H_i^t}
\right)
\leq
\frac{Q}{\eta}t(t+1)$, which is equivalent to $\exp\!\left(
\gamma\frac{D_i^t}{H_i^t}
\right)
\leq
n(100n+1)t(t+1)$.
Taking logarithms yields
\[
\frac{D^t_i}{H^t_i}\leq \frac{4}{3}
\log\!\bigl(
n(100n+1)t(t+1)
\bigr),
\]
and let $C_t$ denote the right hand side of the above inequality.
By argument similar to the proof of Claim~\ref{claim::D^t_i}, one can verify that $D_i^t=\theta V_i^t-B_i^t$.
Combining the identity with the upper bound on $D_i^t$ gives $B_i^t \geq \theta V_i^t-C_tH_i^t$.
Choose a good among $g_1,\ldots,g_t$ whose value to agent $i$ is $H_i^t$. If agent $i$ received such a good, then $B_i^t\geq H_i^t$.
If she did not receive such a good, then $U_i^t=H_i^t$ and hence $B_i^t+U_i^t\geq H_i^t$.
Therefore, we have
\[
\theta V_i^t
\leq
B_i^t+C_tH_i^t\leq
(C_t+1)(B_i^t+U_i^t),
\]
which implies
$\frac{n(B_i^t+U_i^t)}{V_i^t} \geq \frac{n\theta}{C_t+1}$.
Since $n\theta=\frac{1}{4}$, we obtain
\[
\frac{n(B_i^t+U_i^t)}{V_i^t}
\geq
\frac{1}{
4+\frac{16}{3}
\log\!\bigl(
n(100n+1)t(t+1)
\bigr)
}.
\]
Since $n\geq 2$ and $t\geq 1$, one can verify that 
\[
\frac{1}{
4+\frac{16}{3}
\log\!\bigl(
n(100n+1)t(t+1)
\bigr)
}
\geq
\frac{3}{184\log(nt)},
\]
which completes the proof.
\end{proof}

Therefore, one deterministic allocation rule attains the \(\frac{1}{2}\) guarantee under accurate MIV predictions and the same asymptotic order of guarantee as Algorithm~\ref{alg:no-informaton} under arbitrary predictions.

\section{Conclusion}
\label{sec:conclusion}

We studied the deterministic online allocation of indivisible goods under
PROP1, with and without advance information about the agents' maximum item
values. Without additional information, we showed that no positive
competitive ratio depending only on the number of agents can be
guaranteed. More quantitatively, for a fixed number of agents, our upper
and lower bounds determine the optimal dependence on the number of goods
up to an $O(\log\log m)$ multiplicative factor: every deterministic
algorithm has competitive ratio $O(\frac{\log\log m}{\log m})$, while our
algorithm guarantees $\Omega(\frac{1}{\log m})$. With exact MIV information, we
showed that the situation changes substantially. A reciprocal-potential
algorithm guarantees $1/2$-PROP1, independently of the number of agents,
although our two-agent construction shows that exact PROP1 remains
impossible. 
Finally, we gave a single deterministic algorithm with both consistency and robustness. It guarantees \(\frac{1}{2}\)-PROP1 when all MIV predictions are exact and \(\Omega(\frac{1}{\log nm})\)-PROP1 for arbitrary nonnegative predictions. 
Thus, exploiting accurate MIV predictions does not require giving up the asymptotic order of the guarantee available without additional information.
These results identify a clear information hierarchy: exact
MIV information is sufficient to obtain a constant guarantee, but not
exact PROP1, whereas without additional information the best possible
guarantee must decrease with the number of goods.

Several questions remain open. The first is to close the \(O(\log\log m)\) gap for a fixed number of agents in the setting without additional information, . 
The second is to determine the optimal competitive ratio when MIV predictions are accurate: our results together with the results in Neoh and Teh~\cite{neoh2026closing} place it between $19/30$ and
$1-1/(64\cdot 518)$.
Moreover, it would also be interesting to understand whether randomization, restrictions on the adversary, or quantitative assumptions on the prediction error lead to stronger guarantees.

\paragraph*{AI Disclosure.}
All the results in this manuscript were obtained using GPT-5.6-Sol under
the guidance of the authors. The authors verified the proofs for
correctness and used GPT-5.6-Sol to improve the exposition and simplify
some arguments.

\bibliographystyle{plainurl}
\bibliography{mylibrary}

\end{document}